\documentclass[english,runningheads,a4paper,envcountsame]{llncs}
\usepackage[utf8]{inputenc}
\pdfoutput=1 

\usepackage{amsmath}
\usepackage{amssymb} 
\usepackage{bm} 
\usepackage{mathtools} 
\usepackage{stmaryrd} 
\usepackage[caption=false]{subfig}

\usepackage{xcolor}
\colorlet{darkblue}{blue!50!black}
\colorlet{darkgreen}{green!50!black}

\usepackage[pagebackref]{hyperref}
\hypersetup{
    breaklinks=true,
    colorlinks,
    citecolor=darkgreen,
    linkcolor=violet,
    urlcolor=darkblue,
    pdftitle={Computing Least and Greatest Fixed Points in Absorptive Semirings},
    pdfauthor={Matthias Naaf}
}

\usepackage[capitalise]{cleveref}
\crefname{section}{Sect.}{Sect.}
\Crefname{section}{Section}{Sections}

\usepackage{tikz}
\usetikzlibrary{arrows,arrows.meta,calc,shapes,backgrounds,positioning,decorations,decorations.pathreplacing,patterns}

\usepackage{csquotes} 
\usepackage{suffix} 
\usepackage{microtype}
\usepackage[strings]{underscore} 

\spnewtheorem*{remark*}{Remark}{\itshape}{}

\newenvironment{Definition}{\definition\upshape}{\enddefinition}

\newenvironment{Example}{\example}{\hfill$\lrcorner$\endexample}

\newcommand{\IntroParagraph}[1]{\bigskip\noindent\textbf{#1.}}

\renewcommand{\phi}{\varphi}
\renewcommand{\theta}{\vartheta}
\renewcommand{\emptyset}{\varnothing}
\renewcommand{\epsilon}{\varepsilon}

\DeclareMathOperator{\lfp}{\mathbf{lfp}}
\DeclareMathOperator{\gfp}{\mathbf{gfp}}
\newcommand{\Inf}{\bigsqcap}
\newcommand{\Sup}{\bigsqcup}

\newcommand{\Bool}{\mathbb{B}}

\newcommand{\Nat}{\mathbb{N}}
\newcommand{\Ninf}{{\mathbb N}^{\infty}}
\newcommand{\Sorp}{\mathrm{Sorp}}
\newcommand{\Sinf}{{\mathbb S}^{\infty}}
\newcommand{\Trop}{\mathbb{T}}
\newcommand{\Vit}{\mathbb{V}}
\newcommand{\Luk}{\mathbb{L}}

\newcommand{\NN}{\mathbb{N}}
\newcommand{\RR}{\mathbb{R}}

\newcommand*{\XX}{{\bm X}}
\newcommand*{\EE}{\mathcal{E}}

\newcommand*{\Max}{\mathsf{Maximals}}
\newcommand*{\absorbs}{\succeq}

\renewcommand{\AA}{{\bm A}}

\renewcommand*{\root}{\varepsilon}
\newcommand*{\yield}{\mathsf{yd}}
\newcommand*{\mon}{\mathsf{mon}}
\newcommand*{\var}{\mathsf{var}}
\newcommand*{\Trees}{\mathcal{T}}
\newcommand*{\varcount}[2]{|{#1}|_{#2}}
\newcommand*{\moncount}[3]{|{#1}|_{#2,#3}}
\newcommand{\cut}[2]{\scalebox{1}[1.3]{$\|$}\vphantom{|}_{#1}^{#2}}
\newcommand{\scut}[2]{\|_{#1}^{#2}} 
\newcommand{\bb}{\tup b}
\newcommand{\zero}{\tup 0}
\newcommand{\one}{\tup 1}

\newcommand*{\co}{\colon}
\newcommand*{\tup}[1]{\mathbf{#1}}

\newcommand{\prefixneq}{\sqsubset}

\newcommand{\bigmid}{\;\big|\;}
\newcommand{\Bigmid}{\;\Big|\;}
\newcommand{\biggmid}{\;\bigg|\;}

\newcommand{\eqInfo}[1]{\overset{\clap{\scriptsize #1}}{=}}
\newcommand{\leInfo}[1]{\overset{\clap{\scriptsize #1}}{\le}}
\newcommand{\geInfo}[1]{\overset{\clap{\scriptsize #1}}{\ge}}
\newcommand{\eqNote}{\eqInfo{!}}
\newcommand{\leNote}{\leInfo{!}}
\newcommand{\eqStar}{\eqInfo{$(*)$}}
\newcommand{\leStar}{\leInfo{$(*)$}}
\newcommand{\geStar}{\geInfo{$(*)$}}
\newcommand{\eqLabel}[1]{\eqInfo{$(#1)$}}
\newcommand{\leLabel}[1]{\leInfo{$(#1)$}}
\newcommand{\geLabel}[1]{\geInfo{$(#1)$}}

\WithSuffix\newcommand\eqInfo*[1]{\;\eqInfo{#1}\;}
\WithSuffix\newcommand\leInfo*[1]{\;\leInfo{#1}\;}
\WithSuffix\newcommand\geInfo*[1]{\;\geInfo{#1}\;}
\WithSuffix\newcommand\IffInfo*[1]{\;\IffInfo{#1}\;}
\WithSuffix\newcommand\eqNote*{\;\eqNote\;}
\WithSuffix\newcommand\leNote*{\;\leNote\;}
\WithSuffix\newcommand\eqQuestion*{\;\eqQuestion\;}
\WithSuffix\newcommand\eqStar*{\;\eqStar\;}
\WithSuffix\newcommand\leStar*{\;\leStar\;}
\WithSuffix\newcommand\geStar*{\;\geStar\;}
\WithSuffix\newcommand\eqLabel*[1]{\eqInfo*{$(#1)$}}
\WithSuffix\newcommand\leLabel*[1]{\leInfo*{$(#1)$}}
\WithSuffix\newcommand\geLabel*[1]{\geInfo*{$(#1)$}}

\makeatletter
\def\CSum{%
  \let\CSum@subarg\relax
  \let\CSum@suparg\relax
  \CSum@testsub}

\def\CSum@sub_#1{%
  \def\CSum@subarg{#1}%
  \ifx\CSum@suparg\relax
   \expandafter\CSum@testsup
  \else
  \expandafter\CSum@doit
  \fi}

\def\CSum@sup^#1{%
  \def\CSum@suparg{#1}%
  \CSum@testsub@b
}

\def\CSum@testsup{%
  \let\CSum@suparg\@empty
  \@ifnextchar^\CSum@sup\CSum@testsub}

\def\CSum@testsup@b{%
  \let\CSum@suparg\@empty
  \@ifnextchar^\CSum@sup\CSum@doit}

\def\CSum@testsub{%
  \@ifnextchar_\CSum@sub\CSum@testsup@b}

\def\CSum@testsub@b{%
  \@ifnextchar_\CSum@sub\CSum@doit}

\def\CSum@doit{%
\;\sum_{\mathclap{\CSum@subarg}}^{\mathclap{\CSum@suparg}}\,}
\makeatletter

\makeatletter
\def\RSum{%
  \let\RSum@subarg\relax
  \let\RSum@suparg\relax
  \RSum@testsub}

\def\RSum@sub_#1{%
  \def\RSum@subarg{#1}%
  \ifx\RSum@suparg\relax
   \expandafter\RSum@testsup
  \else
  \expandafter\RSum@doit
  \fi}

\def\RSum@sup^#1{%
  \def\RSum@suparg{#1}%
  \RSum@testsub@b
}

\def\RSum@testsup{%
  \let\RSum@suparg\@empty
  \@ifnextchar^\RSum@sup\RSum@testsub}

\def\RSum@testsup@b{%
  \let\RSum@suparg\@empty
  \@ifnextchar^\RSum@sup\RSum@doit}

\def\RSum@testsub{%
  \@ifnextchar_\RSum@sub\RSum@testsup@b}

\def\RSum@testsub@b{%
  \@ifnextchar_\RSum@sub\RSum@doit}

\def\RSum@doit{%
\smashoperator[r]{\sum_{\RSum@subarg}^{\RSum@suparg}}\,}
\makeatletter

\makeatletter
\def\LSum{%
  \let\LSum@subarg\relax
  \let\LSum@suparg\relax
  \LSum@testsub}

\def\LSum@sub_#1{%
  \def\LSum@subarg{#1}%
  \ifx\LSum@suparg\relax
   \expandafter\LSum@testsup
  \else
  \expandafter\LSum@doit
  \fi}

\def\LSum@sup^#1{%
  \def\LSum@suparg{#1}%
  \LSum@testsub@b
}

\def\LSum@testsup{%
  \let\LSum@suparg\@empty
  \@ifnextchar^\LSum@sup\LSum@testsub}

\def\LSum@testsup@b{%
  \let\LSum@suparg\@empty
  \@ifnextchar^\LSum@sup\LSum@doit}

\def\LSum@testsub{%
  \@ifnextchar_\LSum@sub\LSum@testsup@b}

\def\LSum@testsub@b{%
  \@ifnextchar_\LSum@sub\LSum@doit}

\def\LSum@doit{%
\;\smashoperator[l]{\sum_{\LSum@subarg}^{\LSum@suparg}}}
\makeatother

\makeatletter
\def\CProd{%
  \let\CProd@subarg\relax
  \let\CProd@suparg\relax
  \CProd@testsub}

\def\CProd@sub_#1{%
  \def\CProd@subarg{#1}%
  \ifx\CProd@suparg\relax
   \expandafter\CProd@testsup
  \else
  \expandafter\CProd@doit
  \fi}

\def\CProd@sup^#1{%
  \def\CProd@suparg{#1}%
  \CProd@testsub@b
}

\def\CProd@testsup{%
  \let\CProd@suparg\@empty
  \@ifnextchar^\CProd@sup\CProd@testsub}

\def\CProd@testsup@b{%
  \let\CProd@suparg\@empty
  \@ifnextchar^\CProd@sup\CProd@doit}

\def\CProd@testsub{%
  \@ifnextchar_\CProd@sub\CProd@testsup@b}

\def\CProd@testsub@b{%
  \@ifnextchar_\CProd@sub\CProd@doit}

\def\CProd@doit{%
\;\prod_{\mathclap{\CProd@subarg}}^{\mathclap{\CProd@suparg}}\,}
\makeatletter

\makeatletter
\def\RProd{%
  \let\RProd@subarg\relax
  \let\RProd@suparg\relax
  \RProd@testsub}

\def\RProd@sub_#1{%
  \def\RProd@subarg{#1}%
  \ifx\RProd@suparg\relax
   \expandafter\RProd@testsup
  \else
  \expandafter\RProd@doit
  \fi}

\def\RProd@sup^#1{%
  \def\RProd@suparg{#1}%
  \RProd@testsub@b
}

\def\RProd@testsup{%
  \let\RProd@suparg\@empty
  \@ifnextchar^\RProd@sup\RProd@testsub}

\def\RProd@testsup@b{%
  \let\RProd@suparg\@empty
  \@ifnextchar^\RProd@sup\RProd@doit}

\def\RProd@testsub{%
  \@ifnextchar_\RProd@sub\RProd@testsup@b}

\def\RProd@testsub@b{%
  \@ifnextchar_\RProd@sub\RProd@doit}

\def\RProd@doit{%
\smashoperator[r]{\prod_{\RProd@subarg}^{\RProd@suparg}}\,}
\makeatletter

\makeatletter
\def\LProd{%
  \let\LProd@subarg\relax
  \let\LProd@suparg\relax
  \LProd@testsub}

\def\LProd@sub_#1{%
  \def\LProd@subarg{#1}%
  \ifx\LProd@suparg\relax
   \expandafter\LProd@testsup
  \else
  \expandafter\LProd@doit
  \fi}

\def\LProd@sup^#1{%
  \def\LProd@suparg{#1}%
  \LProd@testsub@b
}

\def\LProd@testsup{%
  \let\LProd@suparg\@empty
  \@ifnextchar^\LProd@sup\LProd@testsub}

\def\LProd@testsup@b{%
  \let\LProd@suparg\@empty
  \@ifnextchar^\LProd@sup\LProd@doit}

\def\LProd@testsub{%
  \@ifnextchar_\LProd@sub\LProd@testsup@b}

\def\LProd@testsub@b{%
  \@ifnextchar_\LProd@sub\LProd@doit}

\def\LProd@doit{%
\;\smashoperator[l]{\prod_{\LProd@subarg}^{\LProd@suparg}}}
\makeatother

\newcommand{\Rplus}{+_\RR}
\newcommand{\vvv}[3]{{\scriptsize\begin{pmatrix}#1 \\ #2 \\ #3\end{pmatrix}}}
\newcommand{\Fmaps}{\xmapsto{\!F\!}}
\newcommand{\Fmapsback}{\,\begin{tikzpicture}[baseline]
    \draw [|/.tip={Bar[width=.8ex,round]}] (0,1.2ex)
    edge [|->,out=0,in=0,looseness=3]
    node [above,yshift=1pt,font=\scriptsize] {$\!F$}
    (0,-.5ex);
\end{tikzpicture}}

\tikzset{
    short/.style={ shorten >=#1, shorten <=#1 },
    vertex/.style={draw=black,circle,minimum size=4mm,inner sep=0pt},
    arr/.style={->,>=stealth',short=2pt},
    annot/.style={font=\scriptsize,above},
}

\begin{document}

\title{Computing Least and Greatest Fixed Points\\ in Absorptive Semirings}

\author{Matthias Naaf}
\authorrunning{M. Naaf}
\institute{RWTH Aachen University, Germany\\\email{naaf@logic.rwth-aachen.de}\\}

\maketitle


\begin{abstract}
We present two methods to algorithmically compute both least and greatest solutions
of polynomial equation systems over absorptive semirings (with certain completeness and continuity assumptions), such as the tropical semiring.
Both methods require a polynomial number of semiring operations,
including semiring addition, multiplication and an infinitary power operation.

Our main result is a closed-form solution for least and greatest fixed points based on the fixed-point iteration.
The proof builds on the notion of (possibly infinite) derivation trees;
a careful analysis of the shape of these trees allows us to collapse
the fixed-point iteration to a linear number of steps.
The second method is an iterative symbolic computation in the semiring of generalized absorptive polynomials, largely based on results on Kleene algebras.
\end{abstract}



\section{Introduction}

A recent line of research on semiring provenance analysis for databases \cite{GreenKarTan07,DeutchMilRoyTan14,GreenTan17}, logic \cite{GraedelTan17,DannertGraNaaTan21} and games \cite{GraedelTan20} has identified the class of absorptive, commutative semirings as an appropriate domain for provenance semantics of fixed-point logics \cite{DannertGraNaaTan21} and games with fixed-point semantics, such as Büchi or parity games.
The underlying idea is to replace the Boolean evaluation of formulae by computations in certain semirings.
From this point of view, a formula is essentially a polynomial expression over some semiring, and fixed-point formulae evaluate to least or greatest solutions of polynomial equation systems.
To guarantee the existence and meaningfulness (when interpreted as provenance information) of these fixed points, one assumes that the semiring is equipped with a natural order that is a complete lattice (for the existence) and that the semiring is absorptive, that is, $1 + a = 1$ for all elements $a$.
Absorption guarantees a duality of the semiring operations in the sense that addition is increasing, with least element $0$, while multiplication is decreasing, with greatest element $1$, and it is this property that leads to meaningful provenance information of greatest fixed points \cite{DannertGraNaaTan21}.

This raises the question how one can (efficiently) compute least and greatest solutions of polynomial equation systems over such semirings.
The textbook approach is the fixed-point iteration: start by setting all indeterminates to the smallest (or greatest) semiring value, then repeatedly evaluate the equations to obtain new values for all indeterminates.
In the Boolean setting, this terminates in at most $n$ steps on $n$ indeterminates (due to monotonicity),
but we are also interested in larger and especially infinite semirings such as the tropical semiring\footnote{We use $\Rplus$ for the addition on $\RR$ to distinguish it from the semiring operation $+$.} $\Trop = (\RR_{\ge 0} \cup \{\infty\}, \min, \Rplus, \infty, 0)$.
Several techniques have been developed to compute least solutions. For $\omega$-continuous semirings (where suprema exist and are compatible with the semiring operations),
Hopkins and Kozen \cite{Kleene} have defined a faster iteration scheme based on differentials and more recently, Esparza, Kiefer and Luttenberger \cite{Newton} have used this idea to generalize Newton's method to $\omega$-continuous semirings.
This works surprisingly well for a wide variety of semirings (in fact, their results for idempotent semirings subsume our result for least fixed points).
Gondran and Minoux \cite{GondranMinoux08} use quasi-inverses of elements and matrices to compute least solutions of linear systems and univariate polynomial equations over dioids.
This applies to absorptive semirings, where elements have the trivial quasi-inverse $a^* = 1$ and hence quasi-inverses of matrices always exist.

Our goal is to complement the results in \cite{Kleene,Newton} by also computing \emph{greatest} solutions, as our motivation stems from semiring provenance where both least and greatest fixed points are considered.
To this end, we work with absorptive, \emph{fully} continuous semirings (requiring continuity for both suprema and infima).

\IntroParagraph{Example}
Consider the following graph whose edges are annotated by cost values in the tropical semiring.
A natural example of a greatest fixed point is the minimal cost of an infinite path.
This corresponds to the greatest solution of the equation system given on the right, where each node is represented by an indeterminate and costs appear as coefficients (notice that the right-hand sides are indeed polynomial expressions in terms of the semiring operations).

\begin{center}
\begin{tikzpicture}[node distance=1.3cm, baseline=(b.north)]
\node [baseline,vertex] (a) {$a$};
\node [baseline,vertex,right of=a] (b) {$b$};
\node [baseline,vertex,right of=b] (c) {$c$};
\path [arr]
    (a) edge [loop above] node [annot] {$1$} ()
    (c) edge [loop above] node [annot] {$0$} ()
    (b) edge node [annot] {$1$} (a)
    (b) edge node [annot] {$20$} (c);
\end{tikzpicture}
\hspace{1cm}
$\begin{aligned}
X_a &= 1 \Rplus X_a \\
X_b &= \min(1 \Rplus X_a, 20 \Rplus X_c) \\
X_c &= 0 \Rplus X_c
\end{aligned}$
\end{center}

When we speak of \emph{least} or \emph{greatest} solutions, we always refer to the natural order of the semiring.
In the case of the tropical semiring, this is the inverse of the standard order, so $\infty <_\Trop 20 <_\Trop 1 <_\Trop 0$.
While the least solution of the above system is trivially $X_a = X_b = X_c = \infty$,
the fixed-point iteration for the greatest solution is infinite:
\[
    \vvv 0 0 0 \mapsto \vvv 1 1 0 \mapsto \vvv 2 2 0 \mapsto \vvv 3 3 0 \mapsto \cdots \mapsto
    \vvv {20} {20} 0 \mapsto \vvv {21} {20} 0  \mapsto \vvv {22} {20} 0 \mapsto \vvv {23} {20} 0
    \mapsto \cdots
\]
and converges to the greatest solution: $X_a = \infty$, $X_b = 20$ and $X_c = 0$.

\IntroParagraph{Main Result}
The essential idea to compute such solutions is that greatest fixed points are composed of two parts: a cyclic part that is repeated indefinitely (the loop at $a$ or $c$) and a reachability part to get to the cycle (the edges from $b$).
As both parts can consist of at most $n$ nodes, all information we need is already present after $n$ steps of the fixed-point iteration; we can use this information to abbreviate the iteration.
The formal proof of this observation is based on (infinite) derivation trees, inspired by the derivation trees in the analysis of Newton's method \cite{Newton} and infinite strategy trees in \cite{DannertGraNaaTan21}.
We show that these trees provide an alternative description of the fixed-point iteration; a careful analysis of the shape of the derivation trees then leads to our main result:

\begin{theorem}
Let $F$ be the operator induced by a polynomial equation system in $n$ indeterminates over an absorptive, fully-continuous, commutative semiring.
We can compute in a polynomial number of semiring operations:
\begin{itemize}
\item the least solution: $F^n(\zero)$,
\item the greatest solution: $F^n( \, F^n(\one)^\infty \, )$.
\end{itemize}
\end{theorem}

Here, $a^\infty$ is the \emph{infinitary power} operation $a^\infty \coloneqq \Inf_{n \in \NN} a^n$ which is well-defined (and usually easy to compute) in the absorptive semirings we consider.
For instance, in the tropical semiring we have $0^\infty = 0$ and $a^\infty = \infty$ for $a \neq 0$.

\IntroParagraph{Symbolic Approach}
Our second approach is a technique to eliminate indeterminates one by one, based on the work of Hopkins and Kozen on Kleene algebras \cite{Kleene}.
We apply their symbolic approach to the semiring $\Sinf[\XX]$ of generalized absorptive polynomials, which is perhaps the most relevant semiring for provenance analysis with fixed points, and extend it to include greatest solutions.

\IntroParagraph{Outline}
This paper is structured as follows:
Section 2 introduces the problem setting, in particular the relevant class of semirings, as well as derivation trees.
Section 3 establishes the connection between derivation trees and the fixed-point iteration, and Section 4 builds on this concept to prove our main result.
The symbolic approach for absorptive polynomials is discussed in Section 5.
There are two appendices: Appendix A contains few auxiliary proofs and Appendix B discusses an axiomatization of the infinitary power operation.

\section{Preliminaries}

This section introduces polynomial equation systems, gives an overview on the semirings we are interested in and introduces the notion of derivation trees.

\subsection{Polynomial Equation Systems}

Throughout the paper, we always fix a finite set $\XX = \{ X_1,\dots,X_\ell\}$ of $\ell$ pairwise different indeterminates.
A monomial over $\XX$ is a product of powers of indeterminates, represented as mapping $m \co \XX \to \Nat$ assigning exponents to the indeterminates.

We use bold symbols to denote tuples: $\tup a = (a_1, \dots, a_\ell)$.
In particular, $\tup 0 = (0, \dots, 0)$ and $\tup 1 = (1, \dots, 1)$.
To simplify the presentation, we often avoid numbered indices.
Given $\XX = \{ X_1, \dots, X_\ell \}$, we instead index tuples by these indeterminates.
That is, for a tuple $\tup a = (a_1,\dots,a_\ell)$ and an indeterminate $X \in \XX$, we write $\tup a_X$ for the entry $a_i$ such that $X_i = X$.

\begin{Definition}\label{defPolynomial}
A \emph{polynomial} $P$ over a semiring $(K,+,\cdot,0,1)$ and indeterminates $\XX$ is a finite formal sum of the form $P = \sum_{i=1}^k c_i \cdot m_i$, where the $m_i$ are pairwise different monomials over $\XX$ and $c_i \in K \setminus \{0\}$ are arbitrary coefficients.

Abusing notation, we write $m \in P$ if there is an $i$ with $m = m_i$, and $c \cdot m \in P$ if additionally $c = c_i$.
We may write $P(X_1,\dots,X_\ell)$ to make the indeterminates explicit.
Then, $P(a_1,\dots,a_\ell) \in K$ is the semiring value obtained by instantiating each indeterminate $X_i$ by $a_i \in K$ and evaluating the resulting expression in $K$.
\end{Definition}

\begin{Definition}\label{defEquationSystem}
A \emph{polynomial equation system} $\EE$ over a semiring $K$ and indeterminates $\XX = \{X_1,\dots,X_\ell\}$ is a family of equations $\EE \co \big(X_i = P_i(X_1,\dots,X_\ell)\big)_{1 \le i \le \ell}$
with polynomials $P_i$ over $\XX$ and $K$.

We associate with $\EE$ the operator $F_\EE \co K^\ell \to K^\ell$ defined by $F_\EE(a_1,\dots,a_\ell)_X = P_X(a_1,\dots,a_\ell)$, for $X \in \XX$.
The least (greatest) solution to $\EE$ is thus the least (greatest) fixed point of $F_\EE$.
We drop the index if $\EE$ is clear from the context.
\end{Definition}

Notice that these are quadratic systems, with the number of equations equal to the number of indeterminates.
We recall the example from the introduction in the tropical semiring (where semiring addition is $\min$ and semiring multiplication is $\Rplus$).
Using $\XX = \{X_a,X_b,X_c\}$, we refer to the polynomial equation system as $(X = P_X)_{X \in \XX}$.
For example, $P_{X_b}$ is the polynomial $\min(1 \Rplus X_a, 20 \Rplus X_c)$ consisting of the two coefficient-monomial pairs $1 \Rplus X_a$ and $20 \Rplus X_c$.

\subsection{Semirings}

\begin{Definition}
A (commutative) \emph{semiring} is an algebraic structure
$(K,+,\cdot,0,1)$, with $0\neq1$,  such that $(K,+,0)$
and $(K,\cdot,1)$ are commutative monoids, $\cdot$
distributes over $+$, and $0\cdot a=a\cdot 0=0$.
It is \emph{idempotent} if $a+a = a$ and \emph{absorptive} if $1+a=1$, for all $a \in K$.

In an idempotent semiring $K$, the \emph{natural order} $\le_K$ is the partial order
with $a \le_K b$ if $a + b = b$, for $a,b \in K$.
We drop the index if $K$ is clear from the context.
\end{Definition}

All semirings considered in this paper are commutative and absorptive (except for $\Ninf$ below).
Absorption (also called \emph{0-closed} or \emph{bounded} \cite{Mohri02})
implies idempotence and is equivalent to $1$ being the $\le_K$-maximal element and to multiplication being decreasing, i.e., $ab \le_K a$ for all $a,b \in K$ (dually to increasing addition).

To guarantee the existence of fixed points, we further require that the natural order is a complete lattice so that suprema $\Sup$ and infima $\Inf$ always exist (with respect to $\le_K$).
In addition, we make a continuity assumption stating that the semiring operations commute with the lattice operations on chains (a chain is a totally ordered set).
This is crucial for most of our proofs, but does not seem to be a strong restriction in practice: all natural examples of complete-lattice semirings we are aware of are in fact also fully continuous (a notable exception are binary relations with union and composition, but the latter is not commutative).

\begin{Definition}
An idempotent semiring $K$ is \emph{fully continuous} if $\le_K$ is a complete lattice and for all $a \in K$, all nonempty chains $C \subseteq K$ and $\circ \in \{+,\cdot\}$,
\[
    \Sup (a \circ C) = a \circ \Sup C
    \quad \text{and} \quad
    \Inf (a \circ C) = a \circ \Inf C.    
\]
A homomorphism $h \co K_1 \to K_2$ on fully-continuous semirings is \emph{fully continuous} if $h(\Sup C) = \Sup h(C)$ and $h(\Inf C) = \Inf h(C)$, for all nonempty chains $C \subseteq K_1$.
\end{Definition}

\begin{remark}\label{remarkContinuitySets}
We note that requiring suprema and infima of chains in fact suffices to guarantee the existence of fixed points.
However, in idempotent semirings this already implies the existence of arbitrary suprema and infima \cite{DannertGraNaaTan21} and is thus equivalent to our assumption of a complete lattice.
For continuity, the situation is more complicated:
compatibility of $+$ and $\cdot$ with suprema of chains implies compatibility with arbitrary suprema (see \cref{Appendix:Proofs}), but the same does, in general, not hold for infima.
\end{remark}

Following the above remark, we see that fully-continuous semirings are similar to quantales.
Indeed, since multiplication is compatible with suprema of arbitrary sets, every absorptive, fully-continuous semiring $K$ induces a quantale $(K, \Sup, \cdot)$ with the top element $1$ as unit.
The main difference to quantales is that we additionally require compatibility of semiring operations with infima (but only of chains).
Another related concept is that of topological dioids in \cite{GondranMinoux08} which requires that both operations are compatible with suprema of countable chains.

Since multiplication is decreasing in absorptive semirings,
powers of an element $a$ form a descending chain $1 \ge a \ge a^2 \ge \dots$ whose infimum we denote by $a^\infty$.

\begin{Definition}\label{defInfpow}
In an absorptive, fully-continuous semiring $K$, the \emph{infinitary power} operation is defined by
$a^\infty \coloneqq \Inf_{n \in \Nat} a^n$, for each $a \in K$,
and $\tup a^\infty \coloneqq (a_1^\infty,\dots,a_\ell^\infty)$ for tuples $\tup a \in K^\ell$.
\end{Definition}

Using continuity of multiplication, one can easily verify the properties $(ab)^\infty = a^\infty b^\infty$, $(a^n)^\infty = a^\infty$ and $(a+b)^\infty = a^\infty + b^\infty$ (see \cite{DannertGraNaaTan19} for details).
We remark that it is usually quite easy to compute the infinitary power.
One can further define infinite sum and product operations on families $(a_i)_{i \in I}$ over $K$ with arbitrary index set $I$.
Summation is simply defined as supremum, products can be defined as infimum over finite subproducts (see \cite[Appendix]{DannertGraNaaTan19}).
Here we only need infinite products over finite domain $\{a_i \mid i \in I\}$ (as the polynomials we consider have finitely many coefficients), which are commutative, associative and commute with fully-continuous  homomorphisms and the infinitary power.

Fully-continuous homomorphisms further preserve fixed points of monotone functions, in particular least and greatest solutions of polynomial systems:

\begin{lemma}\label{lemHomoPreservesSolutions}
Let $h \co K_1 \to K_2$ be a fully-continuous homomorphism on absorptive, fully-continuous semirings.
Let $\EE \co (X_i = P_i)_{1 \le i \le n}$ be a polynomial equation system over $K_1$.
Let $h(\EE) \co (X_i = h(P_i))_{1 \le i \le n}$ result from $\EE$ by applying $h$ to all coefficients.
Then, $\lfp(F_{h(\EE)}) = h(\lfp(F_\EE))$ and $\gfp(F_{h(\EE)}) = h(\gfp(F_\EE))$.
\end{lemma}
\begin{proof}
Recall that the semiring operations are fully continuous in $K_1$ and $K_2$.
Hence $F_\EE$ and $F_{h(\EE)}$ are fully continuous as well, and by Kleene's fixed-point Theorem and the continuity of $h$, we get
\[
    h(\lfp(F_\EE)) =
    h\Big( \Inf_{n \in \NN} F_{\EE}^n(\tup 0) \Big) =
    \Inf_{n \in \NN} h\big(F_{\EE}^n(\tup 0) \big) \eqStar
    \Inf_{n \in \NN} \big(F_{h(\EE)}^n(\tup 0) \big) =
    \lfp(F_{h(\EE)}),
\]
where $(*)$ is easy to see by induction, since $h$ is a homomorphism and the operators are defined by polynomials. The proof for greatest solutions is symmetric.
\qed
\end{proof}

\subsubsection*{Examples.}

Some examples of absorptive, fully-continuous semirings are:
\begin{itemize}
\item The \emph{Boolean semiring} $\Bool=(\{0,1\},\vee,\wedge,0,1)$ is the habitat of
logical truth.
\item $\Trop=(\mathbb{R}_{\ge0}^{\infty},\min,\Rplus,\infty,0)$
is the \emph{tropical} semiring used for cost computations.
\item The \emph{Viterbi} semiring $\Vit=([0,1],\max,\cdot,0,1)$
is isomorphic to $\Trop$ and can be used to model confidence scores.
\item The \emph{\L ukasiewicz} semiring $\Luk = ([0,1], \max, \star, 0, 1)$ with $a \star b = \max(0, a+b-1)$, used in many-valued logics.
\item The \emph{min-max} semiring on a totally ordered set $(A,\leq)$
with least element $a$ and greatest element $b$
is the semiring $(A,\max,\min,a,b)$.
\item The semiring of \emph{generalized absorptive polynomials} $\Sinf[\XX]$, defined below.
\end{itemize}

We write $\Ninf$ for the semiring of natural numbers extended by a special element $\infty$ (with $n \cdot \infty = n + \infty = \infty$, for $n \neq 0$).
It is neither absorptive nor idempotent (but fully continuous w.r.t\ the standard order on natural numbers).

\subsubsection*{Absorptive Polynomials.}

The most important absorptive, fully-continuous semiring, both from a provenance perspective and for our proofs, is the semiring of \emph{(generalized\,\footnote{$\Sinf[\XX]$ generalizes the semiring $\Sorp(\XX)$ of absorptive polynomials in \cite{DeutchMilRoyTan14} by adding the exponent $\infty$ (which is needed for \emph{fully}-continuous homomorphisms in \cref{thmSinfUniversal}). We only use $\Sinf[\XX]$ in this paper and hence drop \emph{generalized} in the following.}\!) absorptive polynomials} $\Sinf[\XX]$.
We briefly summarize its definition and key properties from \cite{DannertGraNaaTan21}.
Given a finite set $\XX$ of indeterminates, a (generalized) monomial over $\XX$ is a mapping $m \co \XX \to \Ninf$ (here we also allow the exponent $\infty$),
multiplication adds exponents and the neutral element is $1 \colon X  \mapsto 0$.
We say that a monomial $m_1$ \emph{absorbs} $m_2$, denoted $m_1 \absorbs m_2$, if $m_1(X) \le m_2(X)$ for all $X \in \XX$ (notice that absorption is the \emph{inverse} of the pointwise order on the exponents).
In order to mimic the algebraic property of absorption, polynomials are \emph{antichains} of monomials (which are always finite). Addition and multiplication are defined as usual, but we drop monomials that are absorbed after each operation.
For example, $(XY^2 + X^2Y) \cdot X^\infty = X^\infty Y^2 + X^\infty Y = X^\infty Y$.

\begin{Definition}
The semiring $(\Sinf[\XX], +, \cdot, 0, 1)$ of \emph{(generalized) absorptive polynomials} consists of all antichains of monomials (w.r.t.\ absorption). We write $0$ for the empty antichain and $1$ for the antichain $\{1\}$.
Given $P,Q \in \Sinf[X]$, define
\[
    P + Q = \Max(P \cup Q),
    \quad
    P \cdot Q = \Max\{ m_1 \cdot m_2 \mid m_1 \in P, m_2 \in Q \},
\]
where $\Max(M)$ denotes the set of $\absorbs$-maximal monomials in $M$.
\end{Definition}

This semiring is fully continuous, with $\Sup S = \sum S = \Max(\bigcup S)$ for sets $S$, and absorptive.
Moreover, $\Sinf[\XX]$ is the most general such semiring, as made explicit in the following universal property.
Together with \cref{lemHomoPreservesSolutions}, this is a fruitful tool to simplify reasoning about all absorptive, fully-continuous semirings.

\begin{theorem}[universal property, \cite{DannertGraNaaTan21}]\label{thmSinfUniversal}
Every mapping $h \co \XX \to K$ into an absorptive, fully-continuous semiring $K$
uniquely extends to a fully-continuous semiring homomorphism $h \co \Sinf[\XX] \to K$
(by means of polynomial evaluation).
\end{theorem}

For our technical results, we also need the following observations based on \cite{DannertGraNaaTan19}
(statements (1) and (3) in fact hold in all absorptive, fully-continuous semirings, but the proof is much simpler in $\Sinf[\XX]$, see \cref{Appendix:Proofs}).

\begin{lemma}[\cite{DannertGraNaaTan19}]\label{lemSinfExtra}
Let $S \subseteq \Sinf[\XX]$ and $P \in \Sinf[\XX]$. Then,
\begin{enumerate}
\item $P \cdot \sum S = \sum \{P \cdot Q \mid Q \in S\}$, and
\item $(\sum S)^\infty = \sum \{Q^\infty \mid Q \in S\}$, and
\item $h(\sum S) = \sum h(S)$, if $h \co \Sinf[\XX] \to K$ is a fully-continuous homomorphism.
\end{enumerate}
\end{lemma}

\begin{lemma}[\cite{DannertGraNaaTan19}]\label{lemSinfInfima}
Let $(P_i)_{i \in \Nat}$ be a descending $\omega$-chain with $P_i \in \Sinf[\XX]$.
Then,
\[
    \Inf_{i \in \Nat} P_i = \Sup \Big\{ \Inf_{i \in \Nat} m_i \Bigmid
    \begin{array}{l}
        \text{$(m_i)_{i \in \Nat}$ is a descending $\omega$-chain } \\
        \text{of monomials with $m_i \in P_i$}
    \end{array}
    \Big\}.
\]
\end{lemma}

To clearly distinguish between indeterminates in polynomial equation systems and absorptive polynomials, we often use the indeterminate set $\AA = \{ A_1,\dots,A_k \}$ for the latter, in particular when we use values from $\Sinf[\AA]$ as coefficients.

\subsection{Derivation Trees}

Inspired by the analysis of Newton's method \cite{Newton}, we use derivation trees to describe the behaviour of polynomial equation systems.
For the intuition behind this notion, think of a polynomial system as a formal grammar:
The indeterminates are the nonterminal symbols, coefficients the terminal symbols, and each monomial in $P_X$ gives rise to a production rule for $X$.
We essentially consider derivation trees of this grammar in the usual sense, except that we ignore the order of children (we use commutative semirings) and allow infinite derivations.

While our notion of derivation trees is conceptually identical to \cite{Newton}, we should note that the definition of the yield labeling is slightly different.

\begin{Definition}\label{defTree}
A \emph{derivation tree} $T = (V, E, \var, \yield)$ over a semiring $K$ and indeterminates $\XX$ is a (possibly infinite) tree $(V,E)$ with node labelings
$\var \co V \to \XX$
and
$\yield \co V \to K$, the \emph{yield} of $v$.
We say that $T$ is \emph{from $X$} if for the root $\root$, we have $\var(\root) = X$.
For convenience, we often write $v \in T$ instead of $v \in V$
and refer to $v$ with $\var(v)=X$ as an occurrence of $X$ in $T$.

We associate with each node the monomial $\mon(v) = \prod_{w \in vE} \var(w)$ composed of its children's indeterminates.
We say that $T$ is \emph{compatible with the system $(X = P_X)_{X \in \XX}$} if for each node,
$\yield(v) \cdot \mon(v) \in P_{\var(v)}$.
The set of all derivation trees from $X$ that are compatible with the system $\EE$ is denoted $\Trees(\EE, X)$.
\end{Definition}

Given an equation system $\EE \co (X = P_X)_{X \in \XX}$ and an indeterminate $X$, a derivation tree $T \in \Trees(\EE, X)$ first chooses from the equation $X = P_X$ a monomial $\mon(\root)$ together with its coefficient $\yield(\root)$.
On the next level, it then makes analogous choices for all indeterminates occurring in $\mon(\root)$, where the exponent specifies how often an indeterminate occurs.
The leaves $v$ of such a derivation tree (if they exist) have $\mon(v) = 1$ and correspond to absolute coefficients in one of the equations.
See \cref{figTreeExample} for an example.
We define the yield of an entire tree as the combined yield of all nodes (recall that we assume $K$ to be absorptive and fully-continuous, so infinite products are well-defined):

\begin{Definition}\label{defYield}
The \emph{yield} of a derivation tree $T = (V, E, \var, \yield)$ over $K$ is the (possibly infinite) product
$\yield(T) = \prod_{v \in V} \yield(v) \in K$.
\end{Definition}

We compare the yields of two trees by counting occurrences of coefficients or, equivalently, of monomials from the polynomial equation system.

\begin{Definition}
Let $\EE$ be an equation system over $\XX$, let $Y \in \XX$ and $m \in P_Y$.
For derivation trees $T \in \Trees(\EE, X)$, we define
\[
    \moncount T m Y = \big|\{ v \in T \mid \text{$\mon(v) = m$, $\var(v) = Y$}\} \big| \in \Ninf
\]
as the number of occurrences of $m \in P_Y$ in $T$.
Notice that we use pairs $(m,Y)$ to unambiguously refer to $m \in P_Y$, as $m$ may also occur in other polynomials of $\EE$.
\end{Definition}

\begin{lemma}[yield comparison]\label{lemYieldComparison}
Given a polynomial system $\EE \co (X = P_X)_{X \in \XX}$ over an absorptive, fully-continuous semiring and trees $T,T' \in \Trees(\EE, X)$,
\begin{itemize}
\item if $\moncount T m Y \ge \moncount {T'} m Y$ for all $m \in P_Y$, $Y \in \XX$, then $\yield(T) \le \yield(T')$,
\item if $|T|_{m,Y} = 0$ implies $|T'|_{m,Y} = 0$ for all $m,Y$, then $\yield(T)^\infty \le \yield(T')^\infty$.
\end{itemize}
\end{lemma}
\begin{proof}
We use the properties of infinite products to group the yields by monomials.
For each $m \in P_Y$, let $c_{m,Y} \in K$ be its coefficient, such that $c \cdot m \in P_Y$. Then,
\begin{align*}
  \yield(T) =
  \prod_{v \in T} \yield(v) =
  \prod_{\substack{Y \in \XX, \\ m \in P_Y}} \prod_{\substack{v \in T,\\\var(v) = Y,\\\mon(v) = m}} \yield(v) =
  \prod_{\substack{Y \in \XX, \\ m \in P_Y}} (c_{m,Y})^{\moncount T m Y},
\end{align*}
and the same applies to $T'$.
The product on the right is finite and by absorption, larger exponents lead to smaller values,
hence $\yield(T) \le \yield(T')$.
If we apply the infinitary power, we similarly get
$
  \yield(T)^\infty =
  \prod_{m,Y} (c_{m,Y})^{\infty \cdot \moncount T m Y},
$
where $\infty \cdot \moncount T m Y$ is either $\infty$ (if $\moncount T m Y > 0$) or $0$, implying the second statement. \qed
\end{proof}

\section{Derivation Trees and the Fixed-Point Iteration}

As a first step towards our main result, this section shows that we can express least and greatest solutions in terms of the yields of derivation trees.
Notice that a single derivation tree does not correspond to a solution of the equation system, but only to (the derivation of) a single term in the solution.
We thus consider the sum over all derivation trees.

For least solutions, this was already shown (for a slightly different notion of derivation trees) in \cite{Newton}.
Here we are mostly concerned with the proof for greatest solutions, as this is much more involved due to the trees being infinite.

\begin{theorem}\label{thmFixpointsAsTrees}
Let $K$ be an absorptive, fully-continuous semiring.
Let $\EE \co (X = P_X)_{X \in \XX}$ be a polynomial equation system over $K$.
Then for each $X \in \XX$,
\[
\lfp(F_\EE)_X = \sum_{\substack{T \in \Trees(\EE, X),\\T \text{ is finite}}} \yield(T), \qquad
\gfp(F_\EE)_X = \sum_{T \in \Trees(\EE, X)} \yield(T).
\]
\end{theorem}

We recall that summation is equivalent to supremum (in idempotent semirings).
Here and in the following, we use summation in reminiscence of the general, non-idempotent case (cf.\ \cite{Newton}) and only switch to supremum as needed.
Towards a proof, we first observe that it suffices to prove \cref{thmFixpointsAsTrees} for the most general semiring $K = \Sinf[\AA]$. That is, with the coefficients being absorptive polynomials (not to be confused with the polynomials of the equation system $\EE$).

\begin{claim}
If \cref{thmFixpointsAsTrees} holds for $K = \Sinf[\AA]$, then it also holds for any absorptive, fully-continuous semiring $K$.
\end{claim}
\begin{proof}
Let $K$ be an absorptive, fully-continuous semiring.
Given $\EE \co (X = P_X)_{X \in \XX}$ over $K$,
we construct a symbolic abstraction $\EE' \co (X = P'_X)_{X \in \XX}$ over $\Sinf[\AA]$.
To this end, let $P'_X$ result from $P_X$ by replacing all coefficients with pairwise different indeterminates from $\AA$.
Let $h \co \AA \to K$ be the corresponding instantiation of these indeterminates that reverses this process, so that $h(\EE') = \EE$.
By \cref{thmSinfUniversal}, this mapping induces a fully-continuous homomorphism $h \co \Sinf[\AA] \to K$, so by \cref{lemHomoPreservesSolutions}, we have for each $X$,
\[
  \gfp(F_\EE)_X = \gfp(F_{h(\EE')})_X \eqInfo{(\ref{lemHomoPreservesSolutions})} h(\gfp(F_{\EE'}))_X = h\bigg( \RSum_{T' \in \Trees(\EE',X)} \; \yield(T') \bigg).
\]

Notice that the structure of derivation trees $\Trees(\EE,X)$ only depends on the monomials occurring in $\EE$, but not on the coefficients.
Thus, the derivation trees for $\EE$ and $\EE'$
are identical up to the labeling $\yield$.
Given a tree $T \in \Trees(\EE,X)$, it holds in particular that
\[
  \yield(T) =
  \prod_{v \in T} \yield(v) =
  \prod_{v \in T'} h(\yield(v)) =
  h\Big( \prod_{v \in T'} \yield(v) \Big) =
  h\big( \yield(T') \big),
\]
where $T' \in \Trees(\EE',X)$ is the tree corresponding to $T$
(so that only $\yield$ is changed according to the coefficients in $\EE'$).
By using the one-to-one correspondence between trees $T \in \Trees(\EE,X)$ and $T' \in \Trees(\EE',X)$, we can conclude
\[
  \gfp(F_\EE)_X =
  h\bigg( \RSum_{T' \in \Trees(\EE',X)} \yield(T') \bigg) \eqInfo*{(\ref{lemSinfExtra})}
  \RSum_{T' \in \Trees(\EE',X)} h\big(\yield(T')\big) = 
  \RSum_{T \in \Trees(\EE,X)} \yield(T).
\]
The proof for $\lfp(\EE)$ is symmetric. \qed
\end{proof}

For the remaining section, we fix a polynomial equation system $\EE \co (X = P_X)_{X \in \XX}$ over $K=\Sinf[\AA]$ and consider the induced operator $F$.
The proof proceeds by induction on the fixed-point iterations $F^n(\tup 0)$ and $F^n(\tup 1)$, but requires some preparation.
The idea is that if complete derivation trees correspond to the fixed points, their prefixes should correspond to the steps of the iteration.
These prefixes are defined by simply cutting off the derivation trees at a certain depth and assigning a specific yield to the nodes at the cut-off depth
(eventually, we will simply assign $0$ for the least and $1$ for the greatest fixed point).

\begin{Definition}
Let $T = (V, E, \var, \yield) \in \Trees(\EE,X)$, $n \in \Nat$ and $\bb \in K^\ell$.
Let $V_{\le n} \subseteq V$ be the nodes at depth $\le n$.
We define the \emph{$(n,\bb)$-truncation} of $T$ as
\[
  T \cut n \bb \coloneqq (V_{\le n}, \, E \cap V_{\le n}^2, \, \var, \, \yield'),
  \quad
  \yield'(v) = \begin{cases}
    \bb_{\var(v)}, &\text{$v$ at depth $n$,} \\
    \yield(v), &\text{otherwise.}
  \end{cases}
\]
This defines a derivation tree (compatible with $\EE$ except for its leaves) and we define $\mon(v)$ and $\yield(T \cut n {\bb})$ as in \cref{defTree,defYield} (cf.\ \cref{figTreeExample}).
\end{Definition}

\begin{figure}[t]
\centering
\begin{minipage}{.25\linewidth}
\begin{align*}
    X_1 &= a X_1 + b X_2 X_3 \\
    X_2 &= c X_1^2 \\
    X_3 &= d \\
\end{align*}
\end{minipage}
\hspace{.8cm}
\begin{minipage}{.61\linewidth}
\begin{tikzpicture}[baseline,font=\small,node distance=.75cm]
\node [baseline] (1) {$X_1/b$};
\node [below of=1,xshift=-.6cm] (2) {$X_2/c$};
\node [below of=1,xshift=+.6cm] (2b) {$X_3/d$};
\node [below of=2,xshift=-.55cm] (3) {$X_1/a$};
\node [below of=2,xshift=+.55cm] (3b) {$X_1/a$};
\node [below of=3] (4) {$X_1/a$};
\node [below of=3b] (4b) {$X_1/a$};
\path[draw,->]
(1) edge (2) (1) edge (2b)
(2) edge (3) (2) edge (3b)
(3) edge (4) (3b) edge (4b)
;
\draw [densely dotted,thick] (4) to ($(4)-(0,.5cm)$);
\draw [densely dotted,thick] (4b) to ($(4b)-(0,.5cm)$);
\node [anchor=base,yshift=-1.3em] at (1.north -| 3.west) {$T:\strut$};
\end{tikzpicture}
\hspace{.8cm}
\begin{tikzpicture}[baseline,font=\small,node distance=.75cm]
\node [baseline] (1) {$X_1/b$};
\node [below of=1,xshift=-.6cm] (2) {$X_2/c$};
\node [below of=1,xshift=+.6cm] (2b) {$X_3/d$};
\node [below of=2,xshift=-.55cm] (3) {$X_1/e_1$};
\node [below of=2,xshift=+.55cm] (3b) {$X_1/e_1$};
\path[draw,->]
(1) edge (2) (1) edge (2b)
(2) edge (3) (2) edge (3b)
;
\node [anchor=base,yshift=-1.3em] at (1.north -| 3.west) {$T \cut 2 {\tup z}:\strut$};
\end{tikzpicture}
\hfill
\end{minipage}
\caption[A derivation tree and its (2,e)-truncation]{A derivation tree $T$ and its $(2,\tup e)$-truncation for a sample equation system, with node labels $\var(v)/\yield(v)$.
The trees have yield $\yield(T) = a^\infty bcd$ and $\yield(T \cut 2 {\tup e}) = e_1^2 bcd$.}
\label{figTreeExample}
\end{figure}
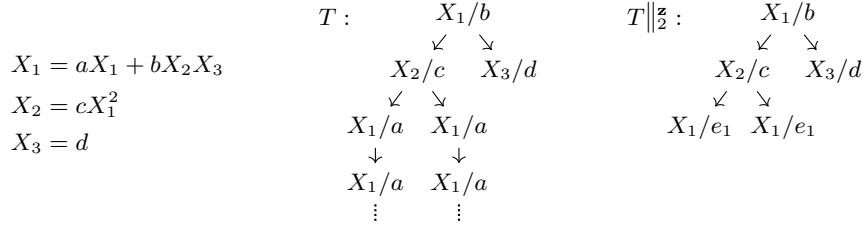

The following, mostly technical lemma establishes the general connection between truncations of derivation trees and the fixed-point iteration.

\begin{lemma}[tree iteration]\label{lemTreeIteration}
Let $K = \Sinf[\AA]$ and $\bb \in K^\ell$. Then,
$F^n(\bb)_X = \sum_{T \in \Trees(\EE,X)} \yield(T \cut n {\bb})$, for all $n \in \Nat$, $X \in \XX$.
\end{lemma}

\begin{proof}
Induction on $n$.
For $n=0$, we trivially have $F^n(\bb)_X = \yield(T \cut 0 \bb) = \bb_X$ for all derivation trees $T \in \Trees(\EE, X)$.
For the induction step, assume that $F^n(\tup b)_X = \sum_{T \in \Trees(\EE,X)} \yield(T \cut n \bb)$ for all $X \in \XX$.
We have to show that
\begin{align*}
  F^{n+1}(\bb)_X = P_X(F^n(\bb)) \eqNote \sum_{T \in \Trees(\EE,X)} \yield(T \cut {n+1} {\bb}).
\end{align*}
To simplify notation, let $\tup a$ be the tuple with $\tup a_X = \sum_{T \in \Trees(\EE,X)} \yield(T \cut n \bb)$.
We can rewrite the left-hand side as follows (recall that $m(X)$ denotes the exponent of $X$ in $m$):
\begin{align*}
  P_X(\tup a) &=
  \sum_{c \cdot m \in P_X} (c \cdot m)(\tup a) =
  \sum_{c \cdot m \in P_X} \Big( c \cdot \prod_{X \in \XX} \tup a_X^{m(X)} \Big) \\ &=
  \sum_{c \cdot m \in P_X} c \cdot \big( \underbrace{\tup a_{X_1} \cdot \tup a_{X_1} \cdots \tup a_{X_1}}_{m(X_1) \text{ times}} \;\cdot\; \dots \;\cdot\; \underbrace{\tup a_{X_\ell} \cdots \tup a_{X_\ell}}_{m(X_\ell) \text{ times}} \big) = \,\dots
\end{align*}
Notice that the unfolded product is finite, since $m$ only has finite exponents.
By \cref{lemSinfExtra}, multiplication distributes over the (infinite) sums $\tup a_{X_1}, \tup a_{X_2},\dots$.
The product $\tup a_{X_1} \cdot \tup a_{X_1} \cdots \tup a_{X_\ell}$ can thus be rewritten as sum:
\begin{align*}
  &= \sum \bigg\{ c \cdot \CProd_{1 \le i \le \ell} \yield(T_{i,1} \cut n \bb) \cdots \yield(T_{i,m(X_i)} \cut n \bb) \biggmid
  \begin{array}{l}
    \text{$c \cdot m \in P_X$ and all possible} \\
    \text{choices of trees $T_{i,j} \in \Trees(\EE,X_i)$}
  \end{array}\bigg\} \\ &=
  \sum \Big\{ \yield(T \cut {n+1} \bb) \bigmid
  \text{$c \cdot m \in P_X$, $T \in \Trees(\EE,X)$ with $\mon(\root) = m$, $\yield(\root) = c$ }\Big\} \\ &=
  \sum \Big\{ \yield(T \cut {n+1} \bb) \bigmid \text{$T \in \Trees(\EE,X)$} \Big\}.
\end{align*}

For these last three steps, recall that derivation trees from $X$ first choose a monomial (and corresponding coefficient) $c \cdot m \in P_X$.
The root $\root$ then has yield $c$ and the children are derivation trees from the indeterminates occurring in $m$.
By commutativity of the infinite product $\yield(T \cut {n+1} \bb)$,
we can group together the yields of the child subtrees,
thereby obtaining a one-to-one correspondence with the product terms in the first step. \qed
\end{proof}

To prove \cref{thmFixpointsAsTrees}, all that is left to do is to consider the supremum of the iteration $F^n(\tup 0)$ and the corresponding tree truncations, and dually the infimum of $F^n(\tup 1)$.
For the infimum, one last obstacle needs to be resolved:
We must show that whenever we pick for each $n$ some $n$-truncation, their infimum can still be realized as yield of an actual (infinite) tree, even if we pick a different tree to truncate for each $n$.
A similar observation has been used for strategy trees of model-checking games in \cite{DannertGraNaaTan21}, where it was called \emph{puzzle lemma} (due to a more involved construction of the infinite tree).

\begin{lemma}[puzzle lemma \cite{DannertGraNaaTan21}]
\label{lemTreesPuzzle}
Let $X \in \XX$ and $K=\Sinf[\AA]$.
Let $(T_n)_{n \in \Nat}$ be a family of trees $T_n \in \Trees(\EE,X)$
such that their yields $\yield(T_n \cut n \one)$ form a descending chain.
Then there is a tree $T' \in \Trees(\EE,X)$ with $\yield(T') \ge \Inf_n \yield(T_n \cut n \one)$.
\end{lemma}

The proof in our setting is quite similar, so we refer to \cite{DannertGraNaaTan21} or \cref{Appendix:Proofs} for a complete proof.
Essentially, the finite number of indeterminates in $\AA$ allows us to choose a sufficiently large $n$ such that the truncation $T_n \cut n \one$ contains a \enquote{nice} part that we can repeat to obtain the infinite tree $T'$.
Here, \enquote{nice} means that no matter how often we repeat this part, the yield does not fall below $\Inf_n \yield(T_n \cut n \one)$.

With this taken care of, we can prove that the sum of all (finite) derivation trees gives the least and greatest solutions.

\begin{proof}[of \cref{thmFixpointsAsTrees}]
Recall that $F$ is fully continuous, as it is defined by polynomial expressions over a fully-continuous semiring.
By Kleene's fixed-point theorem, we can thus express its least (or greatest) fixed point as supremum of $F^n(\zero)$ (or infimum of $F^n(\one)$) over $n \in \Nat$.
By idempotence, sums coincide with suprema, so for the least solution we immediately obtain: \begin{align*}
  \lfp(F)_X =
  \Sup_{n \in \Nat} F^n(\zero)_X &\eqInfo*{(\ref{lemTreeIteration})}
  \Sup_{n \in \Nat} \Big( \RSum_{T \in \Trees(\EE,X)} \yield(T \cut n \zero) \Big) =
  \sum_{T \in \Trees(\EE,X)} \Big( \Sup_{n \in \Nat} \yield(T \cut n \zero) \Big).
\end{align*}
Now observe that $\yield(T \cut n \zero) = \yield(T)$ if $T$ has height $< n$,
otherwise $\yield(T \cut n \zero) = 0$.
Hence $\Sup_n \yield(T \cut n \zero) = \yield(T)$ if $T$ is finite and $0$ otherwise.

It remains to consider the greatest solution.
We apply \cref{lemSinfInfima} to express the infimum in $\Sinf[\AA]$ as a supremum:
\begin{align*}
  \gfp(F)_X &=
  \Inf_{n \in \Nat} F^n(\one)_X \,\eqInfo{(\ref{lemTreeIteration})}
  \Inf_{n \in \Nat} \Big( \RSum_{T \in \Trees(\EE,X)} \yield(T \cut n \one) \Big) \\ &\eqInfo{(\ref{lemSinfInfima})} \;
  \Sup \Big\{ \Inf_{n \in \Nat} y_n \Bigmid
    \begin{array}{l}
      \text{\small $(y_n)_{n \in \Nat}$ is a descending chain of monomials}\\
      \text{\small with $y_n = \yield(T_n \cut n \one)$ for some $T_n \in \Trees(\EE,X)$}
    \end{array}
  \Big \} \\ &\eqInfo{(\ref{lemTreesPuzzle})} \;
  \Sup \Big\{ \yield(T') \bigmid T' \in \Trees(\EE,X) \Big\} = \RSum_{T \in \Trees(\EE,X)} \yield(T).
\end{align*}
In the last line, we apply the puzzle lemma.
This gives us for each monomial chain $(y_n)_{n \in \Nat}$ an infinite tree $T'$ with $\yield(T') \ge \Inf_n y_n$.
Conversely, each tree $T'$ induces the monomial chain defined by $y_n = \yield(T' \cut n \one)$.
It is easy to see that this chain has infimum $\yield(T')$, so we have equality. \qed
\end{proof}

\section{Closed Form Solution}

This section is devoted to the proof of our main result:

\newcounter{backuptheorem}
\setcounter{backuptheorem}{\value{theorem}}
\setcounter{theorem}{0}
\begin{theorem}\label{thmClosedSolution}
Let $K$ be an absorptive, fully-continuous semiring.
Let $\EE \co (X = P_X)_{X \in \XX}$ be a polynomial equation system over $K$ and $\XX = \{X_1,\dots,X_\ell\}$ with induced operator $F_\EE \co K^\ell \to K^\ell$.
Then,
\[
    \lfp(F_\EE) = F_\EE^\ell(\zero),\quad
    \gfp(F_\EE) = F_\EE^\ell( \, F_\EE^\ell(\one)^\infty \,).
\]
\end{theorem}
\setcounter{theorem}{\value{backuptheorem}}

Towards the proof, we again fix a polynomial equation system $\EE \co (X = P_X)_{X \in \XX}$ over an absorptive, fully-continuous semiring $K$ with induced operator $F$.
Recall that $\ell$ denotes the number of equations (and indeterminates) of $\EE$.
Our strategy is to prove that we can always find derivation trees of a certain shape, and that the yield of all other derivation trees is absorbed by these trees.

\subsection{Deterministic Derivation Trees}

\begin{Definition}
A derivation tree $T = (V, E, \var, \yield)$ is said to be \emph{deterministic} if $\mon(v)$ depends only on $\var(v)$.
That is, the indeterminate labels of the (unordered) children of a node $v$ are determined by the node's indeterminate label so that the relation $\{ (\var(v), \mon(v)) \mid v \in V \}$ is a function.
\end{Definition}

To reason about $\gfp(F)$, we must reason about and construct infinite derivation trees.
This is straight-forward for deterministic trees.

\begin{lemma}[deterministic construction]\label{lemDeterministicConstruction}
Let $\XX_0 \subseteq \XX$.
For each $X \in \XX_0$, let $m_X$ be a monomial with $m_X \in P_X$ such that all indeterminates occurring in $m_X$ are contained in $\XX_0$.
Then for each $X \in \XX_0$, there is a deterministic tree $T \in \Trees(\EE, X)$ with $\var(v) \in \XX_0$ and $\mon(v) = m_{\var(v)}$ for all nodes $v \in T$.
\end{lemma}
\begin{proof}[sketch]
Starting from the root $\var(\root)=X$, define the (possibly infinite) tree $T$ inductively by repeatedly adding to each leaf $v$ child nodes according to $m_{\var(v)}$, always maintaining the desired property for all inner nodes. \qed
\end{proof}

It is easy to see that deterministic trees are uniquely defined by their prefix up to depth $\ell-1$, as at depth $\ell$ each path must either end or start to repeat (recall that there are only $\ell$ different indeterminates).
Moreover, once we consider the infinitary power $\yield(T)^\infty$, it does not matter how often a particular coefficient $c$ occurs in $T$, since $(c^n)^\infty = c^\infty$ for all $n > 0$.
This leads to the following simple but essential observations.

\begin{lemma}\label{lemDeterministicPrefix}
If $T \in \Trees(\EE,X)$ is deterministic, every indeterminate that occurs in $T$ also occurs in the truncation $T \cut {\ell-1} \one$.
It follows that $\yield(T)^\infty = \yield(T \cut \ell \one)^\infty$.
\end{lemma}

\begin{corollary}\label{lemDeterministicExists}
For each $T \in \Trees(\EE,X)$, there is a deterministic tree $T' \in \Trees(\EE, X)$ such that $\yield(T \cut \ell \one)^\infty \le \yield(T')^\infty$.
\end{corollary}
\begin{proof}[sketch]
Choose any way to determinize $T \cut \ell \one$ (by \cref{lemDeterministicConstruction}) using only monomials appearing in $T \cut \ell \one$.
This is always possible, as $T \cut \ell \one$ contains at most $\ell$ indeterminates and hence every path must contain a repetition or end in a leaf (cf.\ \cref{Appendix:Proofs}).
The inequality holds by \cref{lemYieldComparison} (yield comparison). \qed
\end{proof}

\subsection{Constructing Simple Trees}

\begin{figure}[t]
\tikzset{sdot/.style={draw=black, fill=black, circle, inner sep=0pt, minimum size=4pt}}
\centering
\subfloat{\label{figTreeConstructionW}\centering
\begin{tikzpicture}[font=\scriptsize,baseline,scale=0.9]
\fill [fill=white] (0,0) -- (1.2,-3) -- (-1.2,-3) -- cycle;
\draw [line join=bevel] (1.2,-3) -- (0,0) -- (-1.2,-3);

\coordinate (X1) at (-.4,-1.6);
\coordinate (X2) at (.45,-1.8);
\coordinate (Y) at (0,-1.4);

\node (0) at (0,0) {\phantom{$X$}};
\node (1) at (X1) {\underline{$X$}};
\node (2) at (Y) {\underline{$X$}};
\node (3) at (X2) {\underline{$Y$}};
\node (4) at (-.2,-2.2) {$X$};
\node (5) at (.7,-2.8) {$X$};
\node (6) at (-.6,-2.7) {$Y$};
\node (7) at (.4,-2.4) {$Y$};
\node (8) at (-.2,-2.75) {$X$};

\path[densely dotted,short=-2pt]
(0) edge (1) (0) edge (2) (0) edge (3)
(1) edge (6) (1) edge (4) (4) edge (8)
(3) edge (7) (3) edge (5)
;

\node [font=\small,anchor=north east] at (-.5,0) {(a)};
\end{tikzpicture}
}
\hfill
\subfloat{\label{figTreeConstructionS}\centering
\begin{tikzpicture}[font=\small,baseline,scale=0.9,xscale=1.1]
\coordinate (X1) at (-.4,-1.6);
\coordinate (X2) at (.45,-1.8);
\coordinate (Y) at (0,-1.4);

\fill [fill=gray!30] ($(X1)+(0,.3)$) -- ({$(X1)-(.6,0)$} |- {0,-3}) -- ({$(X1)+(.0,0)$} |- {0,-3}) -- cycle;
\draw [black] ({$(X1)+(.0,0)$} |- {0,-3}) -- ($(X1)+(0,.3)$) -- ({$(X1)-(.6,0)$} |- {0,-3});

\fill [fill=gray!30] ($(Y)+(0,.3)$) -- ({$(Y)-(.3,0)$} |- {0,-3}) -- ({$(Y)+(.25,0)$} |- {0,-3}) -- cycle;
\draw [black] ({$(Y)+(.25,0)$} |- {0,-3}) -- ($(Y)+(0,.3)$) -- ({$(Y)-(.3,0)$} |- {0,-3});

\fill [fill=gray!30] ($(X2)+(0,.3)$) -- ({$(X2)-(.1,0)$} |- {0,-3}) -- ({$(X2)+(.5,0)$} |- {0,-3}) -- cycle;
\draw [black] ({$(X2)+(.5,0)$} |- {0,-3}) -- ($(X2)+(0,.3)$) -- ({$(X2)-(.1,0)$} |- {0,-3});

\draw [thick, fill=white, line join=bevel] (0,0) -- (-.65,-1.6) -- (-.25,-1.6) -- (-.1,-1.4) -- (.1,-1.4) -- (.25,-1.8) -- (.7,-1.8) -- cycle;

\node [anchor=base,yshift=1pt,font=\scriptsize] at (X1) {$X$};
\node [anchor=base,yshift=1pt,font=\scriptsize] at (Y) {$X$};
\node [anchor=base,yshift=1pt,font=\scriptsize] at (X2) {$Y$};

\node [anchor=south,inner sep=1pt,font=\scriptsize] at (-.63,-3) {$T^\infty_X$};
\node [anchor=south,inner sep=1pt,font=\scriptsize] at (0,-3) {$T^\infty_X$};
\node [anchor=south,inner sep=1pt,font=\scriptsize] at (.62,-3) {$T^\infty_Y$};

\node [font=\small,anchor=north east] at (-.5,0) {(b)};
\end{tikzpicture}
}
\hfill
\subfloat{\label{figTreeConstructionPrefix}\centering
\begin{tikzpicture}[font=\small,baseline,scale=0.9]
\coordinate (root) at (0,0);
\coordinate (droproot) at (-.35,-.4375);
\coordinate (keeproot) at (-.8,-1.0);
\coordinate (X2) at (.7,-1.5); 
\coordinate (X1) at (-.7,-1.5); 
\coordinate (X3) at (0,-1.5); 
\tikzset{crossout/.style={preaction={fill=black!70}, pattern=north east lines, pattern color=gray!60}}
\tikzset{crossout/.style={preaction={fill=gray!15}, pattern=north east lines, pattern color=gray!80}}

\fill [gray!30] ($(X2)+(0,.3)$) -- ({$(X2)-(.1,0)$} |- {0,-3}) -- ({$(X2)+(.5,0)$} |- {0,-3}) -- cycle;
\draw [black] ({$(X2)+(.5,0)$} |- {0,-3}) -- ($(X2)+(0,.3)$) -- ({$(X2)-(.1,0)$} |- {0,-3});

\fill [gray!30] ($(X1)+(0,.3)$) -- ({$(X1)-(.5,0)$} |- {0,-3}) -- ({$(X1)+(.2,0)$} |- {0,-3}) -- cycle;
\draw [black] ({$(X1)+(.2,0)$} |- {0,-3}) -- ($(X1)+(0,.3)$) -- ({$(X1)-(.5,0)$} |- {0,-3});

\fill [crossout] ($(X3)+(0,.3)$) -- ({$(X3)-(.25,0)$} |- {0,-3}) -- ({$(X3)+(.25,0)$} |- {0,-3}) -- cycle;

\fill [white] (0,0) -- (-1.2,-1.5) -- (1.2,-1.5) -- cycle;

\fill [crossout] (droproot) -- (-1.2,-1.5) -- (0.5,-1.5) -- cycle;

\draw [draw=black, thick] (droproot) -- (root) -- (1.2,-1.5) -- (0.5,-1.5) -- cycle;

\draw [draw=black, thick, fill=white] (keeproot) -- (-1.2,-1.5) -- (-0.4,-1.5) -- cycle;

\node [sdot, label={east:\textcolor{black}{\scriptsize$\!Z$}}] (r) at (droproot) {};
\node [sdot, label={west:\scriptsize $Z\!$}] (v) at (keeproot) {};

\draw [thick,->,>=stealth',shorten <=2pt,shorten >=2pt] (v) edge [out=130, in=190, looseness=1.6] (r);

\node [font=\small,anchor=north east,overlay] at (-1.1,0) {(c)};
\end{tikzpicture}
}
\hfill
\subfloat{\label{figTreeConstructionResult}\centering
\begin{tikzpicture}[font=\small,baseline,scale=0.9]
\coordinate (X1) at (-.3,-1.0);
\coordinate (Y) at (0.4,-1.2);

\fill [fill=gray!30] ($(X1)+(0,.3)$) -- ({$(X1)-(.6,0)$} |- {0,-3}) -- ({$(X1)+(.15,0)$} |- {0,-3}) -- cycle;
\draw [black] ({$(X1)+(.15,0)$} |- {0,-3}) -- ($(X1)+(0,.3)$) -- ({$(X1)-(.6,0)$} |- {0,-3});

\fill [fill=gray!30] ($(Y)+(0,.3)$) -- ({$(Y)-(.25,0)$} |- {0,-3}) -- ({$(Y)+(.5,0)$} |- {0,-3}) -- cycle;
\draw [black] ({$(Y)+(.5,0)$} |- {0,-3}) -- ($(Y)+(0,.3)$) -- ({$(Y)-(.25,0)$} |- {0,-3});

\draw [thick, fill=white, line join=bevel] (0,0) -- (-.65,-1.0) -- (0,-1.0) -- (.1,-1.2) -- (.7,-1.2) -- cycle;

\node [anchor=base,yshift=1pt,font=\scriptsize] at (X1) {$X$};
\node [anchor=base,yshift=1pt,font=\scriptsize] at (Y) {$Y$};

\node [anchor=south,inner sep=1pt,font=\scriptsize] at (-.47,-3) {$T^\infty_X$};
\node [anchor=south,inner sep=1pt,font=\scriptsize] at (.52,-3) {$T^\infty_Y$};

\node [font=\small,anchor=north east] at (-.5,0) {(d)};

\draw[thick,decoration={brace, mirror, amplitude=.15cm}, decorate]
(.9,-1.2) -- (.9,0)
node[pos=0.5,right,anchor=west,align=left,xshift=5pt] {$\ell$};
\end{tikzpicture}
}
\caption{Illustration of the construction steps in the proof of \cref{lemTreeConstruction}.}
\label{figTreeConstruction}
\end{figure}
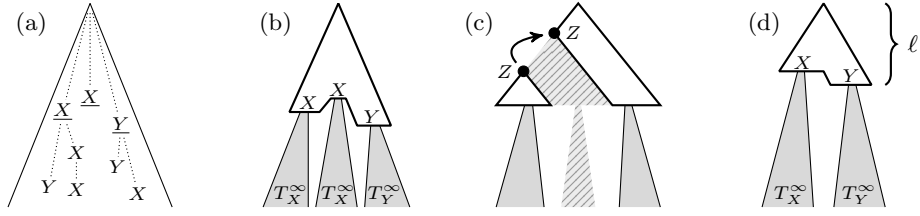

The main insight behind \cref{thmClosedSolution} is that when we sum over the yield of all derivation trees,
it suffices to consider trees of a particular shape corresponding to our intuition from the introduction:
These trees consist of an arbitrary prefix up to (at most) depth $\ell$ (the reachability part), followed by deterministic trees (the cyclic part).
See \cref{figTreeConstructionResult} for an illustration.

\begin{lemma}[main lemma]
\label{lemTreeConstruction}
For each $T \in \Trees(\EE, X)$, there is a derivation tree $T' \in \Trees(\EE,X)$
such that all subtrees rooted at depth $\ell$ in $T'$ are deterministic and use only monomials $m \in P_Y$ (with $Y \in \XX$) that occur infinitely often in $T$.
Moreover, $\moncount {T'} m Y \le \moncount T  m Y$ for all $m \in P_Y, Y \in \XX$.
\end{lemma}

\begin{proof}
Let $\XX_\infty \subseteq \XX$ be the set of indeterminates that occur infinitely often in $T$ (may be empty).
We write $V_X = \{ v \in T \mid \var(v) = X \}$ for the set of nodes labeled $X$.
For each $X \in \XX_\infty$, the set $V_X$ is infinite.
As the polynomial $P_X$ is finite, there must thus be infinitely many $v \in V_X$ with the same monomial $\mon(v)$.
For each $X \in \XX_\infty$, choose such an infinitely often occurring monomial $m_X \in P_X$.
Using \cref{lemDeterministicConstruction}, we obtain for each $X \in \XX_\infty$ a deterministic tree $T^\infty_X \in \Trees(\EE,X)$ such that for all $v \in T^\infty_X$: $\var(v) \in \XX_\infty$ and $\mon(v)$ occurs infinitely often in $T$.

Let $W$ be the set of earliest occurrences of $\XX_\infty$ in $T$ (cf.\ \cref{figTreeConstructionW}). Formally,
$W = \{ v \in T \mid \var(v) \in \XX_\infty$, there is no $v' \prefixneq v$ with $\var(v') \in \XX_\infty \}$.
Now let $S$ be the tree that results from $T$ by replacing the subtree at each $v \in W$ with the tree $T^\infty_{\var(v)}$ (cf.\ \cref{figTreeConstructionS}).
The tree $S$ is almost of the desired shape, but the trees $T^\infty_X$ may be rooted at depth $> \ell$.
To fix this, we consider the prefix up to the subtrees $T^\infty_X$ and eliminate all repetitions of indeterminates within the prefix.
As all indeterminates in the prefix occur only finitely often, we can eliminate repetitions by replacing each first occurrence of an indeterminate $Z$ by a last occurrence of $Z$ within the prefix (cf.\ \cref{figTreeConstructionPrefix}).

More formally, call a path $v_0 v_1 v_2 \dots v_k$ from the root of $S$ \emph{unresolved} if
$k \ge \ell$ and no node on the path is contained in one of the deterministic subtrees $T_X^\infty$.
Given an unresolved path, there must be an indeterminate $Z$ that occurs twice on the path.
Further, $Z \notin \XX_\infty$, as otherwise the nodes labeled $Z$ would lie within one of the deterministic subtrees by construction of $S$ and $W$.
Let $i < j$ be indices such that $v_i$ is the first and $v_j$ the last occurrence of $Z$ on the path, so $\var(v_i) = \var(v_j) = Z$.
Now let $S'$ result from $S$ by replacing the subtree $S_{v_i}$ rooted at $v_i$ with the subtree $S_{v_j}$ rooted at $v_j$, thereby removing at least one occurrence of $Z$ from the tree (\cref{figTreeConstructionPrefix}).

Apply this elimination step exhaustively, until there are no more unresolved paths.
As all indeterminates in the prefix of $S$ occur only finitely often, this process terminates.
Let $T'$ be the resulting tree (notice that $T'$ is not uniquely determined, but this does not affect our argument).
Then $T'$ has the desired shape: when no unresolved path exists, then all nodes at depth $\ell$ (if any) must be contained in one of deterministic subtrees $T_X^\infty$.

Moreover, the elimination step only removes nodes of $S$, but neither adds nodes nor modifies any node labels.
It follows that $\moncount {T'} m Y \le \moncount S m Y$ for all $m \in P_Y, Y \in \XX$.
As the trees $T_X^\infty$ only use monomials that occur infinitely often in $T$, we further have $\moncount S m Y \le \moncount T m Y$ for all $m,Y$, closing the proof.\qed
\end{proof}

\subsection{Proof of the Main Result}

We relate infinite trees of this shape to the expression $F^\ell(\,F^\ell(\one)^\infty \,)$.
The deterministic trees rooted at depth $\ell$ correspond to the inner term $F^\ell(\one)^\infty$, relying on \cref{lemDeterministicPrefix} to ensure that $\ell$ applications of $F$ suffice.
The outer applications of $F$ correspond to the prefix on which we impose no further restrictions (except that it has height at most $\ell$).
The following lemma formalizes this intuition.

\begin{lemma}\label{lemMainProofHardDirection}
Let $\bb$ be the tuple with $\bb_X = \sum_{T \in \Trees(\EE,X)} \yield(T \cut \ell \one)^\infty$ for $X \in \XX$.
For each $T \in \Trees(\EE,X)$, there is a tree $T' \in \Trees(\EE,X)$ such that $\yield(T) \le \yield(T' \cut \ell \bb)$.
\end{lemma}
\begin{proof}
Let $T \in \Trees(\EE,X)$.
Using \cref{lemTreeConstruction}, we obtain a tree $T'$ of a certain shape:
Let $S_1,\dots,S_k$ be the subtrees of $T'$ rooted at depth $\ell$.
These subtrees are deterministic and all monomials occurring in $S_1,\dots,S_k$ occur infinitely often in $T$
and moreover, $\moncount {T'} m Y \le \moncount T m Y$ for all $m \in P_Y, Y \in \XX$.
We claim that
\[
    \yield(T) \leInfo*{!} \yield(T' \cut \ell \one) \cdot \prod_{i=1}^k \yield(S_i)^\infty
\]
To see this, we expand the definition of $\yield$ and rearrange terms.
Borrowing the notation $c_{m,Y}$ for the coefficient of $m \in P_Y$ from the proof of \cref{lemYieldComparison}, we obtain
\begin{align*}
    \prod_{\substack{Y \in \XX, \\ m \in P_Y}} (c_{m,Y})^{\moncount T m Y} \leInfo*{!}
    \prod_{\substack{Y \in \XX, \\ m \in P_Y}} (c_{m,Y})^{\moncount {T' \scut \ell \one} m Y + \sum_{i=1}^k \infty \cdot \moncount {S_i} m Y},
\end{align*}
To prove the inequality, it suffices to show that $\moncount T m Y \ge \moncount {T' \cut \ell \one} m Y + \sum_{i=1}^k \infty \cdot \moncount {S_i} m Y$ for all $m,Y$.
This holds by construction of $T'$:
If $\moncount {S_i} m Y > 0$ for some $i$, then $m$ occurs infinitely often in $T$ and hence $\moncount T m Y = \infty$.
Otherwise, the right-hand side is equal to $\moncount {T' \cut \ell \one} m Y \le \moncount {T'} m Y \le \moncount T m Y$.
This proves our claim.

Now let $v_1,\dots,v_k \in T'$ be the root nodes of the deterministic subtrees $S_1,\dots,S_k$.
By \cref{lemDeterministicPrefix}, 
$\yield(S_i)^\infty = \yield(S_i \cut \ell \one)^\infty \le \bb_{\var(v_i)}$, and thus
\[
    \yield(T) \le
    \yield(T' \cut \ell \one) \cdot \prod_{i=1}^k \yield(S_i)^\infty \le
    \yield(T' \cut \ell \one) \cdot \prod_{i=1}^k \bb_{\var(v_i)} =
    \yield(T' \cut \ell \bb). \tag*{\qed}
\]
\end{proof}

We are now ready to prove our main result.
The statement on the least solution follows rather directly from our earlier considerations.
For greatest fixed points, the previous lemma already proves the difficult direction.

\begin{proof}[of \cref{thmClosedSolution}]
It suffices to consider the case $K = \Sinf[\AA]$ (so that \cref{lemSinfExtra,lemTreeIteration} apply), as the general statement follows with \cref{lemHomoPreservesSolutions}.
We first consider the least solution.
It is clear by monotonicity of $F$ that $F^\ell(\zero)_X \le \lfp(F)_X$.
By \cref{thmFixpointsAsTrees,lemTreeIteration}, it thus suffices to prove
\[
    \lfp(F)_X = \;
    \CSum_{\substack{T \in \Trees(\EE,X)\\T\text{ finite}}} \, \yield(T) \leNote* \,
    \CSum_{T \in \Trees(\EE,X)} \, \yield(T \cut \ell \zero) =
    F^\ell(\zero)_X.
\]
To this end, let $T \in \Trees(\EE,X)$ be finite and obtain $T'$ by \cref{lemTreeConstruction}.
As $T$ is finite, no monomials can occur infinitely often.
Hence $T'$ has no subtrees rooted at depth $\ell$ and is thus of height $< \ell$.
But then, $\yield(T) \le \yield(T') = \yield(T' \cut \ell \zero) \le F^\ell(\zero)_X$.

\medskip
For the greatest solution, we know that $\gfp(F)_X = \sum_{T \in \Trees(\EE,X)} \yield(T)$.
On the other hand, \cref{lemTreeIteration} (tree iteration) entails
\[
  (F^\ell(\one)_X)^\infty =
  \Big(\sum_{T \in \Trees(\EE,X)} \yield(T \cut \ell \one)\Big)^\infty \eqLabel*{\ref{lemSinfExtra}}
  \RSum_{T \in \Trees(\EE,X)} \, \yield(T \cut \ell \one)^\infty.
\]
Now let $\bb = F^\ell(\one)^\infty$. Applying \cref{lemTreeIteration} again gives
\[
  F^\ell(F^\ell(\one)^\infty)_X =
  F^\ell(\bb)_X =
  \RSum_{T \in \Trees(\EE,X)} \, \yield(T \cut \ell \bb).
\]

The direction $\gfp(F) \le F^\ell(F^\ell(\one)^\infty)$ follows immediately from \cref{lemMainProofHardDirection}.
For the other direction, let $T \in \Trees(\EE,X)$.
Let $v_1,\dots,v_k$ be the nodes at depth $\ell$ in $T$.
By distributivity (\cref{lemSinfExtra}), we get
\begin{align*}
  \yield(T \cut \ell \bb) &= \;
  \yield(T \cut \ell \one) \cdot \CProd_{1 \le i \le k} \bb_{\var(v_i)} \\ \;&\eqInfo{dist.} \;\;
  \yield(T \cut \ell \one) \cdot \sum \Big\{ \RProd_{1 \le i \le k} \yield(S_i \cut \ell \one)^\infty
    \bigmid \text{$S_i \in \Trees(\EE,\var(v_i))$ for all $i$} \Big\} \\ &\leLabel{\ref{lemDeterministicExists}} \;\;
   \yield(T \cut \ell \one) \cdot \sum \Big\{ \RProd_{1 \le i \le k} \yield(S_i')^\infty
    \bigmid \text{$S_i' \in \Trees(\EE,\var(v_i))$ for all $i$} \Big\} \\ &\leInfo{abs.} \;\;
   \yield(T \cut \ell \one) \cdot \sum \Big\{ \RProd_{1 \le i \le k} \yield(S_i')
    \bigmid \text{$S_i' \in \Trees(\EE,\var(v_i))$ for all $i$} \Big\} \\ \;&\eqInfo{dist.}\;\;
  \sum \Big\{ \yield(T) \bigmid \text{$T \in \Trees(\EE,X)$} \Big\} =
  \gfp(F)_X. \tag*{\qed}
\end{align*}
\end{proof}

Using this result, we can compute least and, most importantly, greatest solutions of polynomial equation systems in a polynomial number of semiring operations (including the infinitary power).
Notice that, although the proof relied on $\Sinf[\AA]$, the computation happens only in the semiring we consider.
For instance, recall the example in the tropical semiring from the introduction.

\begin{Example}\label{exTropicalClosed}
We recall $X_a = 1 \Rplus X_a$, $X_b = \min(1 \Rplus X_a, 20 \Rplus X_c)$ and $X_c = 0 \Rplus X_c$ from the introduction.
Notice that the one-element of the tropical semiring is the real value $0$.
Using \cref{thmClosedSolution}, we collapse the infinite fixed-point iteration to
\[
    \vvv 0 0 0 \Fmaps
    \vvv 1 1 0 \Fmaps
    \vvv 2 2 0 \Fmaps
    \vvv 3 3 0 \xmapsto{^\infty}
    \vvv \infty \infty 0 \Fmaps
    \vvv \infty {20} 0 \Fmapsback
\]
and obtain the expected solution.
In this example, one iteration of $F$ would actually suffice (instead of $\ell=3$ iterations), since cycles have length one (see the graph in the introduction).
In general, all $\ell$ steps are required (see below).
\end{Example}

Coming back to our original motivation from semiring provenance, we can thus compute semiring provenance information for Büchi games \cite{BuchiTBA} or fixed-point logics such as the modal $\mu$-calculus $L_\mu$ or least fixed-point logic LFP \cite{DannertGraNaaTan21}.
If we only need to compute a polynomial number of fixed points, such as for alternation-free $L_\mu$, this information might be computable in polynomial time -- depending on the cost of the semiring operations.
In the most general semiring $\Sinf[\XX]$, we cannot assume that semiring operations can be performed efficiently, as each multiplication can in the worst case double the number of monomials (this is not avoidable, as one can easily construct equation systems whose solution is an absorptive polynomial describing all exponentially many paths in a graph).
Even if the solution consists of few monomials, we may get an intermediary blowup when computing $F^n(\one)$, as seen in the following example.
In such cases, the symbolic approach presented in Section \ref{sec:Symbolic} may be preferable.

\begin{Example}\label{exLogicClosed}
Consider the polynomial equation system over $\Sinf[a,b,c]$ shown on the left.
(Assuming familiarity with $L_\mu$ and semiring provenance, this results from evaluating $\nu X.\ \Box X \land \Diamond P$ in a semiring interpretation that uses the indeterminates $a,b,c$ to track whether the atom $P$ holds at vertices $v_1,v_2,v_3$.)\\
\begin{minipage}{.5\linewidth}
\begin{align*}
X_1 &= b \cdot X_2 \\
X_2 &= (b+c) \cdot X_2 X_3 \\
X_3 &= a \cdot X_1
\end{align*}
\end{minipage}
\hfill
\begin{minipage}{.35\linewidth}
\begin{tikzpicture}[node distance=1.3cm,every node/.style={font=\scriptsize}]
\node [baseline,vertex] (a) {$v_1$};
\node [baseline,vertex,right of=a] (b) {$v_2$};
\node [baseline,vertex,right of=b] (c) {$v_3$};
\path [arr]
    (c) edge [bend left] (a)
    (b) edge [loop above]()
    (a) edge (b)
    (b) edge (c);
\end{tikzpicture}
\end{minipage}

\vspace{\baselineskip}\noindent
We apply \cref{thmClosedSolution} to compute the greatest solution:
\[
    \one \Fmaps
    \vvv b {b+c} a \Fmaps
    \vvv {b^2+bc} {ab^2+abc+ac^2} {ab} \Fmaps
    \vvv {ab^3+ab^2c+abc^2} {a^{2}b^{4} + a^{2}b^{3}c + a^{2}b^{2}c^{2} + a^{2}bc^{3}} {ab^2+abc} \xmapsto{\!^\infty\!}
    \vvv {a^\infty b^\infty} {a^\infty b^\infty} {a^\infty b^\infty}
    \Fmapsback
\]
Here we need all $\ell=3$ inner applications of $F$ until $a$ appears in the first entry.
We also see that the intermediate polynomials can become much longer than the solution.
(Interpreting $a^\infty b^\infty$ as provenance information, we see that the formula holds at the given graph precisely if $P$ holds at $v_1$ and $v_2$, and it does not matter if it also holds at $v_3$.)
\end{Example}

\begin{remark}
One can generalize our main result to polynomial equation systems that allow $\infty$ as exponent (similar to absorptive polynomials).
We have chosen to spare the reader from the additional complications that arise from the corresponding derivation trees with infinite degree, in particular infinite products on the semiring level, as these are not relevant for the main ideas.
Alternatively, the symbolic approach of \cref{sec:Symbolic} can be used in this setting.
\end{remark}

\section{Symbolic Computation}
\label{sec:Symbolic}

This section complements our main results by a second approach focused specifically on polynomial equation systems over the semiring $\Sinf[\AA]$.
To this end, we adapt results of Hopkins and Kozen on Kleene algebras \cite{Kleene} since we can view absorptive semirings as a special case of Kleene algebras (by setting $a^* = 1$ for all elements $a$).
These results are based on symbolic derivatives of polynomials (which is also the basis for Newton's method \cite{Newton}) to express least solutions.
Here, we generalize this approach to include the infinitary power operation, eventually leading to a similar statement also for greatest solutions.

\begin{remark}
Extensions of Kleene algebra by an operation similar to infinitary power have already been studied in other contexts.
For instance, Cohen \cite{Cohen00} introduces $\omega$-algebras by axiomatizing a unary operation $a^\omega$, which can be interpreted as infinite repetition in $\omega$-words or infinite iteration in the relational model of program analysis.
However, Cohen's axioms seem too weak for the proofs below; we discuss an alternative axiomatization of $a^\infty$ in \cref{Appendix:Infpow}.
\end{remark}

\subsection{Setting and Derivatives}

It is convenient to slightly reformulate our problem setting:
Instead of a system $\EE \co (X = P_X)_{X \in \XX}$ with polynomials $P_X$ over $\XX$ and coefficients $\Sinf[\AA]$, we now regard $P_X$ as an absorptive polynomial $P_X \in \Sinf[\AA \cup \XX]$ (so we no longer distinguish between indeterminates of the polynomial system and indeterminates occurring in coefficients).
This allows a more uniform treatment when we eliminate indeterminates one by one, and it is easy to see that it does not affect the solutions.

To simplify notation, we write $\Sinf[\AA,X]$ for $\Sinf[\AA \cup \{X\}]$.
Recall that we write $P(X) \in \Sinf[\AA, X]$ to make explicit that $X$ may occur in $P$;
then $P(a) \in \Sinf[\AA]$ denotes the polynomial that results from $P(X)$ by replacing $X$ with $a \in \Sinf[\AA]$.
In the following, we use $\XX$ to denote an arbitrary (finite) indeterminate set.

\begin{Definition}
Let $X \in \XX$ and $P(X) \in \Sinf[\XX]$.
We denote the \emph{partial derivative} of $P$ with respect to $X$ as $P'$ (leaving $X$ implicit) and define it inductively by
\begin{align*}
    &X' = 1, \quad
    Y' = 0 \text{ for $X \neq Y \in \XX$}, \\
    &(PQ)' = P' \cdot Q + P \cdot Q', \quad
    (P+Q)' = P' + Q', \quad
    \big({P^\infty}\big)' = P^\infty \cdot P',
\end{align*}
where $P(X),Q(X) \in \Sinf[\XX]$.
\end{Definition}

\subsection{Solutions in One Dimension}

We first show how least and greatest solutions of a single equation $X = P(X)$ can be expressed using derivatives, following the proof in \cite{Kleene}.
Notice that $\Sinf[\XX]$, and in fact any absorptive semiring, can be regarded as a Kleene algebra in the sense of \cite{Kleene} by setting $a^* = 1$ for all $a \in K$.
Hence, most of the lemmas require no modifications, except for our addition of the infinitary power.

\begin{lemma}\label{lemChainTaylor}
Let $P(X),Q(X) \in \Sinf[\XX]$ and further $a,b,c \in \Sinf[\XX]$. Then,
\begin{enumerate}
\item $P(Q)' = P'(Q) \cdot Q'$, \hfill (chain rule, cf.\ \cite{Kleene})
\item $P(a+b) = P(a) + P'(a+b) \cdot b$ \hfill (Taylor's theorem, cf.\ \cite{Kleene}),
\item $ac \le bc \;\implies\; P(a) c \le P(b) c $ \hfill (cf.\ \cite{Kleene}).
\end{enumerate}
\end{lemma}

\begin{proof}
By structural induction on $P$.
The proof of \cite{Kleene} also applies to $\Sinf[\XX]$, so we only have to consider the case $P=H^\infty$.
For the chain rule, we have
\begin{align*}
    \big(H^\infty(Q)\big)' &=
    \big(H(Q)^\infty\big)' =
    H(Q)^\infty \cdot H(Q)' \\ &= 
    H(Q)^\infty \cdot H'(Q) \cdot H' =
    \big(H(Q)^\infty\big)' \cdot Q' =
    \big(H^\infty(Q)\big)' \cdot Q',
\end{align*}
and for Taylor's theorem,
\begin{align*}
    H^\infty(a+b) &=
    H(a+b)^\infty \\ &=
    \big( H(a) + H'(a+b) \cdot b \big)^\infty \\ &\eqStar
    H(a)^\infty + \big(H(a) + H'(a+b) \cdot b \big)^\infty \cdot H'(a+b) \cdot b \\ &=
    H(a)^\infty + H(a+b)^\infty \cdot H'(a+b) \cdot b \\ &=
    H^\infty(a) + H^\infty(a+b) \cdot H'(a+b) \cdot b \\ &=
    H^\infty(a) + (H^\infty)'(a+b) \cdot b.
\end{align*}
In $(*)$, we use the fact that $(a+b)^\infty = a^\infty + (a+b)^\infty \cdot b$ for elements $a,b$ of any absorptive, fully-continuous semiring.
For (3), we apply full continuity:
\begin{align*}
  H^\infty(a)c =
  \Big( \Inf_{n \in \Nat} H^n(a) \Big) \cdot c =
  \Inf_{n \in \Nat} \big( H^n(a) c \big) \le
  \Inf_{n \in \Nat} \big( H^n(b) c \big) = H^\infty(b)c. \tag*{\qed}
\end{align*}
\end{proof}

For the result about the greatest solution, we need an additional observation about the infinitary power operation: 

\begin{lemma}\label{lemInfpowIndProof}
Let $K$ be an absorptive, fully-continuous semiring.
Then all elements $a,b,c \in K$ satisfy:
\[
    c + ab \ge b \quad\implies\quad
    c + a^\infty b \ge b.
\]
\end{lemma}
\begin{proof}
We first show by induction that $c + a^nb \ge b$ for all $n \in \Nat$.
For $n=0$, trivially $c + a^0b = c+b \ge b$, and for $n=1$ this holds by assumption.
The induction step follows from absorption:
\[
  c + a^{n+1}b \;\eqInfo{abs}\; c + ac + a^{n+1}b = c + a(c+a^nb) \ge c + ab \ge b.
\]
The claim then follows by full continuity:
\[
  c + a^\infty b = c + \Big(\Inf_{n \in \Nat} a^n\Big)b = \Inf_{n \in \Nat} (c+a^nb) \ge \Inf_{n \in \Nat} b = b. \tag*{\qed}
\]
\end{proof}

Using the observations in \cref{lemChainTaylor}, Hopkins and Kozen prove that the least solution is $P'(P(0))^* \cdot P(0)$, which in our setting is equal to $P(0)$ and can in fact be derived directly from absorption, without derivatives.
However, using derivatives allows us to also express greatest solutions:

\begin{theorem}\label{thmIterativeOne}
Let $P(X) \in \Sinf[\AA,X]$.
Then $X = P(X)$ has the least solution $P(0)$ and
the greatest solution $P(0) + P'(1)^\infty$ in $\Sinf[\AA]$.
\end{theorem} 
\begin{proof}
See \cite{Kleene} for a proof of the least solution.
Alternatively, let $a_0 \in \Sinf[\AA]$ be the absolute coefficient of $P$ (i.e., the sum of all monomials not containing $X$) so that $P(0) = a_0$.
Since we have $m(a_0) \le a_0$ for every monomial $m$ containing $X$, it follows that $P(P(0)) = P(a_0) = a_0$ by absorption, so $P(0)$ is the least solution.

For the greatest solution, we first prove that $P(0) + P'(1)^\infty$ is a solution to the inequality $X \le P(X)$:
\begin{align*}
    P\big(P(0) + P'(1)^\infty\big) \eqLabel*{\ref{lemChainTaylor}.2}{}
    &P(0) + P'\big(P(0) + P'(1)^\infty\big) \cdot \big(P(0) + P'(1)^\infty\big) \\ \ge{}\;
    &P(0) + P'(P'(1)^\infty) \cdot P'(1)^\infty \\ \geLabel*{\ref{lemChainTaylor}.3}{}
    &P(0) + P'(1) \cdot P'(1)^\infty = P(0) + P'(1)^\infty.
\end{align*}
We next show that this is the greatest solution to $X \le P(X)$.
To this end, let $a \in \Sinf[\AA]$ be a solution, i.e., $a \le P(a)$.
As $1 \ge a$, we get
\[
    P(0) + P'(1) \cdot a \ge P(0) + P'(a) \cdot a \eqLabel*{\ref{lemChainTaylor}.2} P(a) \ge a.
\]
Using \cref{lemInfpowIndProof}, we can conclude
\[
    P(0) + P'(1)^\infty \geInfo*{abs} P(0) + P'(1)^\infty \cdot a \ge a.
\]
Finally, note that the greatest solution to $X \le P(X)$ is also the greatest solution to $X = P(X)$ by the well-known Knaster-Tarski theorem. \qed
\end{proof}

\subsection{Solutions of Larger Systems}

To solve systems with more than one equation, we eliminate indeterminates one by one, in each step applying \cref{thmIterativeOne}.
The main theoretical underpinning is the uniformity of the solutions in one indeterminate.
In \cite{Kleene}, this uniformity follows from the axiomatic proofs and the fact that instantiations preserve the axioms of Kleene algebra.
Here, we instead appeal to the universal property of $\Sinf[\XX]$ and the fact that fully-continuous semiring homomorphisms preserve least and greatest fixed points.

For the sake of simplicity, we only consider systems of two equations; we can inductively apply the same approach to larger systems.
Moreover, we only state the result for greatest solutions, as least solutions are symmetric.
We use the notation $\gfp(X \mapsto P(X,Y))$ to refer to the greatest solution of the equation $X = P(X,Y)$ using \cref{thmIterativeOne} (where we treat the additional indeterminate $Y$ as a coefficient, i.e., we apply the theorem with $Y \in \AA$).
With this notation, we can formulate the solution of a system in two indeterminates as follows:

\begin{theorem}\label{thmIterativeTwo}
Consider the equation system $\EE \co X=P(X,Y), Y=Q(X,Y)$ with $P,Q \in \Sinf[\AA,X,Y]$.
Let further
\begin{align*}
    H(Y) &= \gfp(X \mapsto P(X,Y)) && \in \Sinf[\AA,Y], \\
    b &= \gfp(Y \mapsto Q(H(Y),Y)) && \in \Sinf[\AA].
\end{align*}
Then $(H(b), b)$ is the greatest solution of $\EE$.
\end{theorem}

\begin{proof}
It is easy to see that $(H(b),b)$ is a solution:
By definition of $b$, we have $Q(H(b),b) = b$.
By definition of $H$, we further have $P(H(Y),Y) = H(Y)$, and by applying the instantiation $Y \mapsto b$ we get $P(H(b),b) = H(b)$.

To prove that $(H(b),b)$ is the greatest solution, we make use of the universal property.
Let $(c,d)$ be any solution with $c,d \in \Sinf[\AA]$.
We claim that $H(d) = \gfp(X \mapsto P(X,d))$.
To see this, consider the definition of $H(Y)$ and apply the instantiation $Y \mapsto d$.
By the universal property, this instantiation is fully continuous and thus preserves greatest fixed points.

Now, since $(c,d)$ is a solution, we have $P(c,d) = c$, and since $H(d)$ is the greatest solution to $X = P(X,d)$, we must have $H(d) \ge c$.
Then also $Q(H(d),d) \ge Q(c,d) = d$.
Since $b$ is the greatest solution to $Q(H(Y),Y) = Y$ and hence also to $Q(H(Y),Y) \ge Y$, we have $b \ge d$.
Finally, $H(b) \ge H(d) \ge c$, so $(H(b),b)$ is indeed the greatest solution. \qed
\end{proof}

We can apply this second technique to semirings other than $\Sinf[\AA]$ by first performing a symbolic abstraction.
That is, we replace all coefficients by pairwise different indeterminates from $\AA$, then compute the solution and apply the reverse instantiation (which preserves solutions).

\begin{Example}
Recall our example in the tropical semiring $\Trop$.
By replacing coefficients with indeterminates $a,b,c$, we obtain the equation system on the right.
\begin{align*}
X_a &= 1 \Rplus X_a & X_a &= a \cdot X_a \\
X_b &= \min(1 \Rplus X_a, 20 \Rplus X_c) \mathrlap{\qquad\rightsquigarrow} & X_b &= a \cdot X_a + b \cdot X_c \\
X_c &= 0 \Rplus X_c & X_c &= c \cdot X_c
\end{align*}
We solve the system over $\Sinf[X_a,X_b,X_c,a,b,c]$ by the symbolic approach:
\begin{itemize}
\item
$\gfp(X_a \mapsto a \cdot X_a) = 0 + a^\infty = a^\infty$ \hfill (by \cref{thmIterativeOne}),
\item
$\gfp(X_b \mapsto a \cdot a^\infty + b \cdot X_c) = a^\infty + b \cdot X_c$ \hfill (we first instantiate $X_a$ by $a^\infty$)
\item $\gfp(X_c \mapsto c \cdot X_c) = c^\infty$
\end{itemize}
The greatest solution is thus $X_a = a^\infty$, $X_b = a^\infty + b c^\infty$, $X_c = c^\infty$.
Applying the reverse substitution, we get the expected solution $(\infty, 20, 0)$ in $\Trop$.
\end{Example}

Usually, the closed-form solution in \cref{thmClosedSolution} is preferable, as we can work directly in the target semiring.
The symbolic technique is best suited to compute solutions in $\Sinf[\AA]$, which is of interest for semiring provenance analysis.
Compared to the closed-form solution, we need slightly fewer computation steps and can often avoid an intermediate blowup in the size of the polynomials.

\begin{Example}
Recall the equation system $X_1 = b \cdot X_2$, $X_2 = (b+c) \cdot X_2 X_3$, $X_3 = a \cdot X_1$ of \cref{exLogicClosed}. Eliminating indeterminates one by one, we get
\begin{itemize}
\item $\gfp(X_1 \mapsto b \cdot X_2) = b \cdot X_2$
\item $\gfp(X_2 \mapsto (b+c) \cdot X_2 X_3) = \big((b+c)X_3\big)^\infty = b^\infty X_3^\infty + c^\infty X_3^\infty$
\item $\gfp(X_3 \mapsto ab \cdot (b^\infty X_3^\infty + c^\infty X_3^\infty)) = a^\infty b^\infty + a^\infty b^\infty c^\infty = a^\infty b^\infty$
\end{itemize}
and by substituting the results backwards, we obtain $X_1 = X_2 = X_3 = a^\infty b^\infty$.
\end{Example}

\begin{remark}
An attentive reader may have noticed that we have phrased all results in this section for the semiring $\Sinf[\XX]$ of absorptive polynomials, whereas the results in \cite{Kleene} apply to polynomials $K[\XX]$ over an arbitrary Kleene algebra $K$.
The reason for this restrictive choice is that the iterative lifting to multivariate equation systems in \cref{thmIterativeTwo} requires some form of compositionality:
in \cite{Kleene}, $K[X,Y]$ is viewed as $K[X][Y]$.
Unfortunately, the usual notion of polynomials $K[X]$ over an absorptive semiring $K$, as in \cref{defPolynomial}, is itself not absorptive, and it seems not obvious how one would define an absorptive version of $K[X]$ (for instance, if $a,b \in K$ are incomparable we want $(aX + aX^2) + bX^2 = aX + bX^2$ by absorption, but also $aX + (aX^2 + bX^2) = aX + (a+b)X^2$, violating associativity).
As our proofs rely on absorption, they do not apply to $K[X][Y]$.

Since our main result already provides a direct computation in any absorptive, fully-continuous semiring, we have here restricted our interest to $\Sinf[\XX]$.
In fact, it was this limitation that motivated our search for a closed form solution.
\end{remark}

\begin{remark}\label{remarkQuestion}
For least solutions, the one-dimensional solution in \cref{thmIterativeOne} in fact implies the solution in \cref{thmClosedSolution}, as shown in \cite[Prop.~28]{GGS19}.
It seems an interesting question if a similar connection holds for greatest solutions, i.e., can the solution $F^\ell(F^\ell(\one)^\infty)$ for $\ell$ equations be derived from the solution $P(0) + P'(1)^\infty$ of a single equation by algebraic methods, without derivation trees?
\end{remark}

\section{Conclusion}

We have presented two methods to compute least and, most importantly, greatest solutions of polynomial equation systems over absorptive, fully-continuous semirings.
Both methods require only polynomially many applications of the semiring operations and the infinitary power, in terms of the number of equations.

While we assume full continuity mostly to guarantee the existence of both kinds of solutions, absorption is a strong assumption that leads to a particularly simple way of computing solutions.
Our motivation to consider absorptive semirings comes from semiring provenance of fixed-point logics, where our methods can directly be applied to compute provenance information, for example of Büchi games or formulae in fixed-point logics such as $L_\mu$ or LFP.

The first method, and our main result, is a closed-form solution that works in any absorptive, fully-continuous semiring and is as easy as computing the standard fixed-point iteration with an added application of the infinitary power.
To prove the correctness for greatest solutions, we extended the notion of derivation trees used in the analysis of Newton's method \cite{Newton} to infinite trees.
Derivation trees provide an intuitive tool to understand the fixed-point iteration, but require somewhat involved arguments and constructions to properly handle infinite trees.
Our main technical contribution is that it suffices to consider trees of a particular shape resembling the solution term $F^\ell(F^\ell(\one)^\infty)$, intuitively corresponding to a reachability prefix with infinitely repeating deterministic subtrees.
For the second method, we applied results on least solutions over Kleene algebras \cite{Kleene} specifically to the semiring of generalized absorptive polynomials, and extended these results by similar observations for greatest solutions.

Comparing the two proofs, we see that the symbolic approach has a simpler algebraic proof, raising the question whether we can avoid the constructions of infinite trees in our main proof in favor of algebraic arguments.
A further direction for future work is to study systems of nested fixed points over absorptive semirings.
Recently, quasipolynomial-time algorithms have been developed to solve such systems in the Boolean case \cite{ArnoldNiwinskiParys21} or over finite lattices \cite{HausmannSchroder19}.
With the simple computation based on the fixed-point iteration, absorptive semirings might be a candidate to further increase the applicability of these algorithms.

\subsubsection*{Acknowledgements.}

I would like to thank the anonymous reviewers for their helpful comments and for suggesting related concepts, in particular \cite{Cohen00,GGS19,GondranMinoux08} which lead to \cref{remarkQuestion} and \cref{Appendix:Infpow}.


\bibliographystyle{splncs04}
\renewcommand{\doi}[1]{\url{https://doi.org/#1}} 
\renewcommand{\doi}[1]{\href{https://doi.org/#1}{\texttt{doi:#1}}} 
\bibliography{fullversion}

\begin{thebibliography}{10}
\providecommand{\url}[1]{\texttt{#1}}
\providecommand{\urlprefix}{URL }
\providecommand{\doi}[1]{https://doi.org/#1}

\bibitem{ArnoldNiwinskiParys21}
Arnold, A., Niwiński, D., Parys, P.: A quasi-polynomial black-box algorithm
  for fixed point evaluation. In: Baier, C., Goubault-Larrecq, J. (eds.) 29th
  EACSL Annual Conference on Computer Science Logic (CSL 2021). Leibniz
  International Proceedings in Informatics (LIPIcs), vol.~183, pp. 9:1--9:23.
  Dagstuhl (2021). \doi{10.4230/LIPIcs.CSL.2021.9}

\bibitem{Cohen00}
Cohen, E.: Separation and reduction. In: Backhouse, R.C., Oliveira, J.N. (eds.)
  Mathematics of Program Construction. Lecture Notes in Computer Science,
  vol.~1837, pp. 45--59. Springer (2000). \doi{10.1007/10722010_4}

\bibitem{DannertGraNaaTan19}
Dannert, K., Gr{\"a}del, E., Naaf, M., Tannen, V.: Generalized absorptive
  polynomials and provenance semantics for fixed-point logic. arXiv: 1910.07910
  [cs.LO] (2019), \url{https://arxiv.org/abs/1910.07910}, full version of
  \cite{DannertGraNaaTan21}

\bibitem{DannertGraNaaTan21}
Dannert, K., Gr{\"a}del, E., Naaf, M., Tannen, V.: Semiring provenance for
  fixed-point logic. In: Baier, C., Goubault-Larrecq, J. (eds.) 29th EACSL
  Annual Conference on Computer Science Logic (CSL 2021). Leibniz International
  Proceedings in Informatics (LIPIcs), vol.~183, pp. 17:1--17:22. Dagstuhl
  (2021). \doi{10.4230/LIPIcs.CSL.2021.17}

\bibitem{DeutchMilRoyTan14}
Deutch, D., Milo, T., Roy, S., Tannen, V.: Circuits for datalog provenance. In:
  Proc. 17th International Conference on Database Theory ICDT. pp. 201--212.
  OpenProceedings.org (2014). \doi{10.5441/002/icdt.2014.22}

\bibitem{Newton}
Esparza, J., Kiefer, S., Luttenberger, M.: Newtonian program analysis. Journal
  of the ACM  \textbf{57}(6), ~33 (2010). \doi{10.1145/1857914.1857917}

\bibitem{GGS19}
Ghilardi, S., Gouveia, M.J., Santocanale, L.: Fixed-point elimination in the
  intuitionistic propositional calculus. {ACM} Trans. Comput. Log.
  \textbf{21}(1),  4:1--4:37 (2019). \doi{10.1145/3359669}

\bibitem{GondranMinoux08}
Gondran, M., Minoux, M.: Graphs, dioids and semirings: new models and
  algorithms, Operations Research/Computer Science Interfaces Series, vol.~41.
  Springer (2008). \doi{10.1007/978-0-387-75450-5}

\bibitem{BuchiTBA}
Gr{\"a}del, E., L{\"u}cking, N., Naaf, M.: Strategy analysis in b{\"u}chi games
  by valuations in absorptive semirings. arXiv: 2106.12892 [cs.LO] (2021),
  \url{https://arxiv.org/abs/2106.12892}

\bibitem{GraedelTan17}
Gr{\"a}del, E., Tannen, V.: Semiring provenance for first-order model checking.
  arXiv:1712.01980 [cs.LO] (2017), \url{https://arxiv.org/abs/1712.01980}

\bibitem{GraedelTan20}
Gr{\"a}del, E., Tannen, V.: Provenance analysis for logic and games. Moscow
  Journal of Combinatorics and Number Theory  \textbf{9}(3),  203--228 (2020).
  \doi{10.2140/moscow.2020.9.203}

\bibitem{GreenKarTan07}
Green, T., Karvounarakis, G., Tannen, V.: Provenance semirings. In: Principles
  of Database Systems {PODS}. pp. 31--40 (2007). \doi{10.1145/1265530.1265535}

\bibitem{GreenTan17}
Green, T., Tannen, V.: The semiring framework for database provenance. In:
  Proceedings of PODS. pp. 93--99. {ACM} (2017). \doi{10.1145/3034786.3056125}

\bibitem{HausmannSchroder19}
Hausmann, D., Schr{\"o}der, L.: Computing nested fixpoints in quasipolynomial
  time. arXiv:1907.07020 [cs.CC] (2019), \url{https://arxiv.org/abs/1907.07020}

\bibitem{Kleene}
Hopkins, M., Kozen, D.: Parikh's theorem in commutative {K}leene algebra. In:
  Proceedings. 14th Symposium on Logic in Computer Science. pp. 394--401. IEEE
  (1999). \doi{10.1109/LICS.1999.782634}

\bibitem{Kozen94}
Kozen, D.: A completeness theorem for kleene algebras and the algebra of
  regular events. Inf. Comput.  \textbf{110}(2),  366--390 (1994).
  \doi{10.1006/inco.1994.1037}

\bibitem{Markowsky76}
Markowsky, G.: Chain-complete posets and directed sets with applications.
  Algebra universalis  \textbf{6}(1),  53--68 (1976)

\bibitem{Mohri02}
Mohri, M.: Semiring frameworks and algorithms for shortest-distance problems.
  Journal of Automata, Languages and Combinatorics  \textbf{7}(3),  321--350
  (2002). \doi{10.25596/jalc-2002-321}

\end{thebibliography}


\clearpage
\appendix

\renewcommand{\thelemma}{\thesection\arabic{theorem}}
\renewcommand{\theproposition}{\thesection\arabic{theorem}}
\setcounter{theorem}{0}

\section{Omitted Proofs}
\label{Appendix:Proofs}

This appendix contains proofs that were omitted or only sketched in the main paper.

\subsection{Proofs of Section 2}

In \cref{remarkContinuitySets}, we mentioned the observation in \cite{DannertGraNaaTan21} that the existence of suprema of chains implies the existence of arbitrary surpema (and thus also arbitrary infima so that $\le_K$ is a complete lattice).
This follows from arguments of Markowsky \cite{Markowsky76} on chain-complete posets:
first observe that due to idempotence, the supremum of finitely many elements is simply their sum; suprema of finite sets and of chains then suffice to guarantee suprema of arbitrary sets, as shown in \cite[Corollary 5]{Markowsky76}.
Following Markowsky's arguments, one can further show the following relation between suprema of chains and sets:

\begin{lemma}
Let $K_1,K_2$ be fully-continuous semirings.
If $f \co K_1 \to K_2$ preserves suprema of finite sets and of nonempty chains, then $f$ also preserves suprema of arbitrary sets.
\end{lemma}

This applies in particular to addition and multiplication with a fixed element, as both preserve addition and hence suprema of finite sets.
It follows that in every fully-continuous semiring $K$, it holds that $a \cdot \Sup S = \Sup aS$ for every $a \in K$ and every set $S \subseteq K$.
In other words, $(K, \Sup, \cdot)$ is a quantale.

\begingroup
\renewcommand{\thelemma}{\ref{lemSinfExtra}}
\addtocounter{theorem}{-1}
\begin{lemma}
Let $S \subseteq \Sinf[\XX]$ and $P \in \Sinf[\XX]$. Then,
\begin{enumerate}
\item $P \cdot \sum S = \sum \{P \cdot Q \mid Q \in S\}$, and
\item $(\sum S)^\infty = \sum \{Q^\infty \mid Q \in S\}$, and
\item $h(\sum S) = \sum h(S)$, if $h \co \Sinf[\XX] \to K$ is a fully-continuous homomorphism.
\end{enumerate}
\end{lemma}
\endgroup
\begin{proof}
We recall that summation is the same as supremum in idempotent semirings.
Statement (2) was shown in \cite{DannertGraNaaTan19} and the other statements follow by the same argument: since $\sum S = \Max (\bigcup S)$, we can write $\sum S$ as finite sum $\sum S = m_1 + \dots + m_k$ where each $m_i$ is contained in some $Q_i \in S$.

For (1), we then have
$P \cdot \sum S = Pm_1 + \dots + P m_k \le P Q_1 + \dots + P Q_k \le \sum \{ P Q \mid Q \in S \}$
and for (3),
$h(\sum S) = h(m_1 + \dots + m_k) = h(m_1) + \dots + h(m_k) \le h(Q_1) + \dots + h(Q_k) \le \sum h(S)$.
The other direction follows by monotonicity in both cases, since $\sum S \ge Q$ for all $Q \in S$.
\qed
\end{proof}

\subsection{Proofs of Section 3}

\begingroup
\renewcommand{\thelemma}{\ref{lemTreesPuzzle}}
\addtocounter{theorem}{-1}
\begin{lemma}[puzzle lemma \cite{DannertGraNaaTan21}]
Let $X \in \XX$ and $K=\Sinf[\AA]$.
Let $(T_n)_{n \in \Nat}$ be a family of trees $T_n \in \Trees(\EE,X)$
such that their yields $\yield(T_n \cut n \one)$ form a descending chain.
Then there is a tree $T' \in \Trees(\EE,X)$ with $\yield(T') \ge \Inf_n \yield(T_n \cut n \one)$.
\end{lemma}
\endgroup

\begin{proof}
We sketch the main steps of the proof, more details can be found in \cite{DannertGraNaaTan19}
(notice that our notion of truncation is simpler, as we cut off at a certain depth instead of counting $R$-nodes).

\begin{enumerate}
\item \textit{Chain splitting}

Let $e_n(m,Y) = \varcount {T_n \cut n \one} {m,Y} \in \Ninf$
be the number of occurrences of $m \in P_Y$ in the tree $T \cut n \one$.
As usual, let $c_{m,Y}$ be the coefficient
of $m$ in $P_Y$ so that
\[
  y_n \coloneqq \yield(T_n \cut n \one) =
  \prod_{\substack{Y \in \XX,\\m \in P_Y}} (c_{m,Y})^{e_n(m,Y)}.
\]
Using the chain splitting lemma from \cite{DannertGraNaaTan19},
we can write the infimum as
\[
  \Inf_{n \in \Nat} y_n =
  \prod_{\substack{Y \in \XX,\\m \in P_Y}} (c_{m,Y})^{e(m,Y)}, \qquad \text{with } e(m,Y) \coloneqq \Sup_{n \in \Nat} e_n(m,Y).
\]

\item \textit{Problematic monomials}

We say that a monomial $m \in P_Y$ is \emph{problematic} if $e(m,Y)$ is finite.
Unproblematic monomials do not impose any restrictions,
as they may appear arbitrarily often (finite or infinite) in the tree $T'$ we construct
(and $T'$ still satisfies the desired inequality $\yield(T') \ge \Inf_n y_n$).
As the polynomial equation system is finite, there are only finitely many problematic monomials.

\item \textit{Decomposition into $k$-layers}

We now decompose each of the trees $T_n \cut n \one$ into $k$-layers,
each of which simply consists of $k$ consecutive levels of the tree.
We choose the layers such that they cover the entire tree and do not overlap.
If $k$ is some constant, we can choose a large enough $n$
such that there is a $k$-layer in $T_n \cut n \one$ that does not contain any problematic monomials
(as there are only finitely many).

\item \textit{Repetition of a $k$-layer}

We now construct $T'$ by first following the tree $T_n$ (for the chosen $n$),
but upon reaching the $k$-layer without problematic monomials,
we continue $T'$ by repeating this layer over and over, so that no further problematic monomials occur in $T'$.
To this end, note that the $k$-layer is a forest consisting of several trees.
We determinize each such tree $S$ by \cref{lemDeterministicConstruction} and obtain a (possibly infinite) deterministic tree $S'$ that we use in $T'$ to replace $S$.
To ensure that \cref{lemDeterministicConstruction} can be applied, we choose $k = \ell+1$, so that each path trough the $k$-layer must contain a repetition of indeterminates
(cf.\ \cref{lemDeterministicExists}).

\item \textit{Conclusion}

This construction ensures that the number of occurrences of problematic monomials in $T'$
is bounded by their occurrences in $T_n \cut n \one$.
In other words, $\varcount {T'} {m,Y} \le e_n(m,Y) \le e(m,Y)$ for all problematic $m \in P_Y$
and thus $\varcount {T'} {m,Y} \le e(m,Y)$ for all (problematic and unproblematic) $m,Y$.
Hence
\[
  \yield(T') =
  \prod_{\substack{Y \in \XX,\\m \in P_Y}} (c_{m,Y})^{\varcount {T'} {m,Y}} \ge
  \prod_{\substack{Y \in \XX,\\m \in P_Y}} (c_{m,Y})^{e(m,Y)} =
  \Inf_{n \in \Nat} y_n. \tag*{\qed}
\]
\end{enumerate}
\end{proof}

\subsection{Proofs of Section 4}

\begingroup
\renewcommand{\thecorollary}{\ref{lemDeterministicExists}}
\addtocounter{theorem}{-1}
\begin{corollary}
For each $T \in \Trees(\EE,X)$, there is a deterministic tree $T' \in \Trees(\EE, X)$ such that $\yield(T \cut \ell \one)^\infty \le \yield(T')^\infty$.
\end{corollary}
\endgroup
\begin{proof}
We use \cref{lemDeterministicConstruction} to construct a deterministic tree $T'$ with $\yield(T \cut \ell \one)^\infty \le \yield(T')^\infty$ as follows.
We define $\XX_0$ and monomials $m_X$ for $X \in \XX_0$ inductively by traversing the tree $T$ level by level, starting with the root.
We always maintain the following invariant: after traversing level $k$, all indeterminates occurring in $m_X$ for $X \in \XX_0^{(k)}$ are contained in $\XX_0^{(k)}$ or occur at level $k+1$ in $T$.

For level $0$ (consisting only of the root with $\var(\root) = X$), we set $\XX_0^{(0)} \coloneqq \{X\}$ and $m_X \coloneqq \mon(\root)$.
The invariant holds due to the children of $\root$.
Assume we have processed level $k$.
Let $\XX_{\text{todo}}^{(k)}$ be the set of indeterminates that occur in $m_X$ for $X \in \XX_0^{(k)}$, but are not contained in $\XX_0^{(k)}$.
If $\XX_{\text{todo}}^{(k)} = \emptyset$, we are done.
Otherwise, for each $Y \in \XX_{\text{todo}}^{(k)}$, there is a node $v$ at level $k+1$ with $\var(v) = Y$ by the invariant.
Choose any such node and set $m_Y \coloneqq \mon(v)$.
By definition of derivation trees, all indeterminates of $m_Y$ occur as children of $v$ on level $k+2$, so the invariant holds.
Then set $\XX_0^{(k+1)} \coloneqq \XX_0^{(k)} \cup \XX_{\text{todo}}^{(k)}$.

This process never enters level $\ell$:
On each level, at least one indeterminate is added to $\XX_0^{(k)}$ (or we stop) and there are only $\ell$ different indeterminates.
Hence all monomials $m_X$ we choose occur in the truncation $T \cut \ell \one$.
Now obtain $T'$ by applying \cref{lemDeterministicConstruction} to $\bigcup_k \XX_0^{(k)}$ and the chosen monomials $m_X$. The inequality $\yield(T \cut \ell \one)^\infty \le \yield(T')^\infty$ follows by \cref{lemYieldComparison}.
\end{proof}

\section{On an Axiomatization of the Infinitary Power}
\label{Appendix:Infpow}
\newcommand{\infpow}{{}^\infty}

This appendix discusses an axiomatization of the infinitary power operation $a^\infty$ in absorptive semirings, similar to Cohen's $\omega$-algebras \cite{Cohen00}.
To justify the new axioms, we derive the common properties of the infinitary power and provide a completeness result.

\begin{Definition}\label{defInfpowAxioms}
An \emph{$\infpow$-algebra} is an algebraic structure $(K,+,\cdot,\infpow,0,1)$ such that $(K, +, \cdot, 0, 1)$ is an absorptive semiring and the following axioms hold:
\begin{align*}
a^\infty &= a a^\infty \tag{$\infty$ abs}\\
b \le c + ab \;&\to\; b \le c + a^\infty b \tag{$\infty$ ind}
\end{align*}
where $\le$ denotes the natural order.
\end{Definition}

We remark that every $\infpow$-algebra is in particular a Kleene algebra (see e.g.\ \cite{Kozen94}) by setting $a^* = 1$ for all $a \in K$.
For a complete axiomatization, we can thus use axioms (3) -- (13) of \cite{Kozen94} together with commutativity $ab=ba$, absorption $1 + a = 1$ and the $\infpow$-axioms.
Moreover, every $\infpow$-algebra is also an $\omega$-algebra by setting $a^\omega = a^\infty$,
since ($\infty$ ind) and absorption imply Cohen's axiom
\[
  b \le c + ab \;\to\; b \le a^*c + a^\infty. \tag{$\star$ ind}
\]

Clearly, every absorptive, fully-continuous semiring is an $\infpow$-algebra (by defining $a^\infty$ as in \cref{defInfpow}), since ($\infty$ ind) holds by \cref{lemInfpowIndProof} and ($\infty$ abs) follows from continuity: $a a^\infty = a \cdot \Inf_n a^n = \Inf_n a^{n+1} = a^\infty$.
Conversely, we show how to derive the main properties of the infinitary power in such semirings from the axioms in \cref{defInfpowAxioms}.
The first lemma provides an alternative proof of \cref{lemChainTaylor} (3) and implies that $\infpow$ is monotone w.r.t.\ $\le$ (using $c=1$).

\begin{lemma}
Every $\infpow$-algebra satisfies the implication
\[
  ac \le bc \;\to\; a^\infty c \le b^\infty c.
\]
\end{lemma}
\begin{proof}
Using ($\infty$ abs) and the assumption, we get $a^\infty c = a^\infty a c \le a^\infty bc$.
By absorption, we thus have $a^\infty c \le b(a^\infty c) \le b(a^\infty c) + b^\infty c$.
Applying ($\infty$ ind) and absorption then gives $a^\infty c \le b^\infty (a^\infty c) + b^\infty c = b^\infty c$ as claimed. \qed
\end{proof}

\begin{lemma}\label{lemInfpowProperties}
In every $\infpow$-algebra, $\infpow$ has the following properties:
\begin{enumerate}
\item $a^n \ge a^\infty$ and $a^n a^\infty = a^\infty$, for all $n \in \NN$,
\item $a^\infty a^\infty = a^\infty$,
\item $(a^\infty)^\infty = a^\infty$,
\item $(ab)^\infty = a^\infty b^\infty$,
\item $(a+b)^\infty = a^\infty + b^\infty$.
\end{enumerate}
\end{lemma}
\begin{proof}
\begin{enumerate}
\setlength{\itemsep}{.5em}
\item Follows inductively from ($\infty$ abs) and absorption.

\item Trivially $a^\infty a^\infty \le a^\infty$ by absorption.
For the other direction, we have $a^\infty \le aa^\infty$ by ($\infty$ abs) and thus $a^\infty \le a^\infty a^\infty$ by ($\infty$ ind).

\item By ($\infty$ abs) and absorption, $(a^\infty)^\infty = a^\infty \cdot (a^\infty)^\infty \le a^\infty$.
Conversely, $a^\infty \le a^\infty a^\infty$ by (2) and thus $a^\infty \le (a^\infty)^\infty a^\infty \le (a^\infty)^\infty$ by ($\infty$ ind) and absorption.

\item We have $a^\infty b^\infty \le (ab) a^\infty b^\infty$ by ($\infty$ abs) and thus $a^\infty b^\infty \le (ab)^\infty a^\infty b^\infty \le (ab)^\infty$ by ($\infty$ ind) and absorption.
For the other direction, we observe that $(ab)^\infty = ab(ab)^\infty \le a(ab)^\infty$ by ($\infty$ abs) and absorption.
Hence $(ab)^\infty \le a^\infty (ab)^\infty$ by ($\infty$ ind) and symmetrically, $(ab)^\infty \le b^\infty (ab)^\infty$.
Combining both, we get $(ab)^\infty \le a^\infty b^\infty (ab)^\infty \le a^\infty b^\infty$, with absorption in the last step.

\item By monotonicity, $a+b \ge a,b$ implies $(a+b)^\infty \ge a^\infty, b^\infty$ and thus $(a+b)^\infty \ge a^\infty + b^\infty$ by idempotence (addition is supremum).
Conversely, we have $(a+b)^\infty = a(a+b)^\infty + b(a+b)^\infty$ by ($\infty$ abs).
Thus $(a+b)^\infty \le a^\infty (a+b)^\infty + b^\infty(a+b)^\infty \le a^\infty + b^\infty$ by applying ($\infty$ ind) twice, with absorption in the last step. \qed
\end{enumerate}
\end{proof}

The above lemmas provide a sanity check for our axiomatization.
It is further easy to see that we cannot omit ($\infty$ abs) or ($\infty$ ind), as the remaining axiom would be satisfied by either $a^\infty = a$ or $a^\infty = 0$, which, in general, violates the property $a^\infty = \Inf_n a^n$ we want to axiomatize.
To further justify our axioms, we observe the following straightforward completeness result.

\begin{proposition}
Let $\alpha,\beta$ be two terms constructed from variables, the constants $0$ and $1$, and the operations $+$, $\cdot$ and $\infpow$.
Then, $\alpha = \beta$ holds in every absorptive, fully-continuous semiring if, and only if, it holds in every $\infpow$-algebra.
\end{proposition}
\begin{proof}
One implication is trivial, since every absorptive, fully-continuous semiring is an $\infpow$-algebra.
For the converse, assume that $\alpha = \beta$ holds in every absorptive, fully-continuous semiring.
We show that we can derive $\alpha = \beta$ from the $\infpow$-axioms and the properties of absorptive semirings.

Using \cref{lemInfpowProperties} and distributivity, we can rewrite each term into a sum of monomials (where we allow $\infty$ as exponent).
We can thus derive equalities of the form $\alpha = m_1 + \dots + m_k$ and $\beta = m_1' + \dots + m'_l$.
Let $\XX$ be the set of variables occurring in $\alpha$ and $\beta$.
We can then view $P_\alpha \coloneqq m_1 + \dots + m_k$ and $P_\beta \coloneqq m_1' + \dots + m'_l$ as absorptive polynomials in $\Sinf[\XX]$.
Since $\alpha = \beta$ and hence $P_\alpha = P_\beta$ holds in every absorptive, fully-continuous semiring, it also holds in $\Sinf[\XX]$.
That is, $\Max\{m_1,\dots,m_k\} = \Max\{m_1',\dots,m'_l\}$.

To prove that we can derive $P_\alpha = \sum \Max\{m_1,\dots,m_k\}$, and analogously for $P_\beta$, it suffices to show that whenever $m \absorbs m'$ holds for two monomials in $\Sinf[\XX]$, we can derive $m \ge m'$ from the axioms.
Since $\cdot$ is monotone and $\XX$ finite, we can reason about each variable in $m,m'$ separately.
Now assume that $X^a \absorbs X^b$, so $a \le b$.
If $a,b < \infty$, we can derive $X^a \ge X^b$ by absorption.
If $b = \infty$, then we can derive $X^a \ge X^b$ using \cref{lemInfpowProperties}.
It follows that $P_\alpha = P_\beta$  and hence $\alpha = \beta$ holds in every $\infpow$-algebra. \qed
\end{proof}

\end{document}